\documentclass[12pt,a4paper]{article}
 
\usepackage[a4paper,hmargin={20mm,20mm},vmargin={27mm,27mm}]{geometry}
\usepackage{authblk}

\usepackage{mathrsfs}
\usepackage{amsfonts}
\usepackage{amssymb}
\usepackage{array}
\usepackage{eurosym}
\usepackage{bm}
\usepackage[scheme = plain, fontset = founder]{ctex}

\usepackage{multirow,bigstrut}
\usepackage{makecell}
\usepackage{enumerate}
\usepackage{tikz}
\usepackage{graphicx}
\usepackage{stfloats}
\usepackage{float}
\usepackage{array}
\usepackage{booktabs}
\usepackage{mathtools}
\usepackage{titlesec}
\usepackage[colorlinks,
            linkcolor=red,
            anchorcolor=blue,
            citecolor=blue
            ]{hyperref}

\usepackage[linesnumbered,boxed,ruled,commentsnumbered]{algorithm2e}
\usepackage{algorithmicx}
\usepackage{algpseudocode}

\titleformat*{\section}{\large\bfseries}
\titleformat*{\subsection}{\normalsize\bfseries}

\newenvironment{proof}{\noindent{\em \textbf{Proof.}}}{\quad \hfill$\Box$\vspace{2ex}}

\newtheorem{theorem}{Theorem}[section]

\newtheorem{remark}[theorem]{Remark}

\numberwithin{equation}{section}

\def \T {\mathbb{T}}

\def\qi {\mathbf{i}}

\newcommand{\norm}[1]{\left\lVert#1\right\rVert}
\newcommand{\abs}[1]{\left|#1\right|}

\title{\Large \textbf{FFT reconstruction of signals from MIMO sampled data}}

\author[a]{\small  Dong Cheng\thanks{chengdong720@163.com}}

\author[b]{\small  Xiaoxiao Hu\thanks{huxiaoxiao@wmu.edu.cn}}

\author[c]{\small Kit Ian Kou\thanks{kikou@umac.mo}}

\affil[a]{\small{Department of Mathematics, Faculty of Arts and Sciences, Beijing Normal University, Zhuhai 519087, China}}
\affil[b]{\small{The First Affiliated Hospital of Wenzhou Medical University, Wenzhou Medical University, Wenzhou 325035, China}}
\affil[c]{\small{Department of Mathematics, Faculty of Science and Technology, University of Macau, Macao, China}}

\date{}

\begin{document}
  \maketitle
\begin{abstract}
\normalsize

This paper introduces an innovative approach for signal reconstruction using data acquired through multi-input–multi-output (MIMO) sampling. First, we show that it is possible to perfectly  reconstruct a set of periodic band-limited signals \(\{x_r(t)\}_{r=1}^R\) from the samples of \(\{y_m(t)\}_{m=1}^M\), which are the output signals of a MIMO system with inputs \(\{x_r(t)\}_{r=1}^R\). Moreover, an FFT-based algorithm is designed to perform the reconstruction efficiently. It is demonstrated that this algorithm  encompasses FFT interpolation and multi-channel interpolation as special cases.   Then,  we investigate the consistency property and the aliasing error of the proposed sampling and reconstruction framework   to evaluate its effectiveness in reconstructing non-band-limited signals.  The analytical expression for the averaged mean square error (MSE) caused by aliasing is presented.
Finally,   the theoretical results are validated by  numerical simulations, and the performance of the proposed reconstruction method in the presence of noise is also examined.

\end{abstract}


\begin{msc}
42A15,  94A12,	41A25.
\end{msc}


\section{Introduction}\label{S_intro}

In numerous applications such as imaging technology \cite{liu2017signal,Li2019multichannel} and communications \cite{Chen2013shannon,Solodky2021optimal}, it often happens that convolutional measurements or generalized samples are   taken from systems. A common problem that needs to be dealt with  is how to reconstruct   continuous signals from these convolved samples. It is well known that a fundamental result in this field is the so-called generalized sampling expansion (GSE) proposed by Papoulis \cite{papoulis1977generalized}. GSE indicates that a band-limited signal can uniquely reconstructed  from the samples   taken from multiple output signals   of    linear   time invariant (LTI) systems provided that the LTI systems satisfy certain conditions.   The sampling strategy in  \cite{papoulis1977generalized} is also called multi-channel sampling \cite{Kang2010asymmetric} and has become a crucial topic in signal processing due to its theoretical significance and wide application.

The exploration and analysis of multi-channel sampling    has been broadened in various directions.  For instance, the multi-channel sampling theorem has been extended for   signals band-limited in the sense of  general integral transforms, such as fractional Fourier transform (FrFT) \cite{liu2017signal}, linear canonical transform (LCT) \cite{xu2017multichannel},  offset LCT \cite{wei2019convolution} and free metaplectic transform (FMT) \cite{Azhar2023papoulis}.  To eliminate the   band-limited constraint,   multi-channel sampling expansions were also established     in  shift invariant subspaces \cite{Kang2010asymmetric,zhao2018generalized} and the space of  signals with finite rate of innovation (FRI) \cite{akhondi2010multichannel} as well as the atomic spaces generated by arbitrary unitary operators \cite{Pohl2012U-invariant}. In addition, a consistency analysis of  multi-channel sampling formulas  was also presented in \cite{xu2017multichannel}. These extensions allow us to use the multi-channel sampling methodology more flexibly and conveniently. Another consideration for multi-channel sampling  is how to deal with signals that do not have a precise mathematical deterministic description, i.e., random signals.  The authors in \cite{Prender2006minimum,eldar_2015,medina2018papoulis} studied   the reconstruction formulas of random signals from  multichannel samples. The optimal multichannel sampling schemes were discussed in \cite{Chen2013shannon,Solodky2021optimal} to maximize   channel capacity and minimize    mean square error. In \cite{Cheng2022signal}, some stable multi-channel reconstruction methods were  provided  in the presence of noise by introducing the optimal post-filtering technique. It was shown that the optimal post-filtering technique gives a superior capability than the Wiener filtering in noise reduction.

The above multi-channel sampling systems are single-input and multiple-output (SIMO). In many applications such as    Doppler information estimation \cite{Yu2010MIMO}   and multiuser wireless communications \cite{wang2022nonlinear}, however, the involved  multi-channel systems   have  multiple  sources, thereby multi-input-multi-output (MIMO) channels arise.  The vector sampling expansion (VSE), proposed by Seidner et al. \cite{Seidner1998intro,Seidner2000vector}, considers the reconstruction of the input signals of   MIMO channels and   encompasses GSE as a special case. The sampling strategy of VSE is also called MIMO sampling \cite{Venka2003sampling}. In a series of papers, Venkataramani and Bresler  \cite{Venka2003sampling,Venka2003filter,venka2004multiple} extensively investigated  necessary and sufficient conditions on the   channel and the sampling rate/density that allow   perfect reconstruction of    multi-band input signals.
Like the SIMO case, the MIMO sampling and reconstruction has been extended to the band-limited signals in FrFT domain \cite{Ma2023nonuniform} and LCT domain \cite{Sharma2011vector}. Note that the original VSE was restricted to the same bandwidth in all inputs,  an interesting   generalization of VSE comes from \cite{Feuer2006generalization}, where  the input signals of the MIMO system are allowed to have  different bandwidths.
To perform MIMO sampling and reconstruction on non-band-limited signals, Shang et al. \cite{shang2007vector}  studied  the   vector  sampling problem  in shift-invariant subspaces and gave some necessary and sufficient conditions  for   the establishment of the error-free reconstruction formula.

\begin{figure}
  \centering
  \includegraphics[width=10cm]{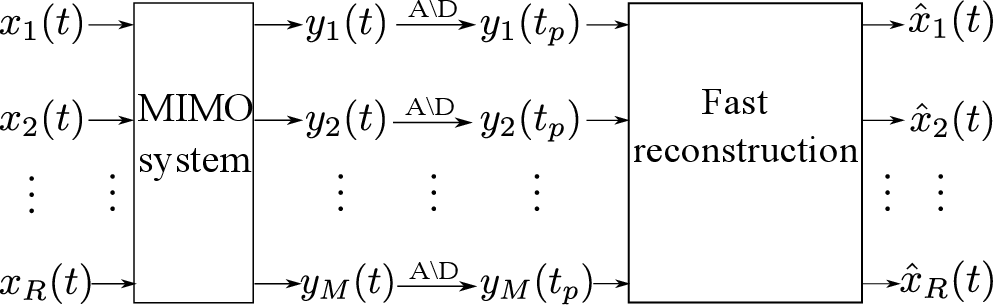}
  \caption{Diagram of MIMO sampling and reconstruction.}\label{MIMOsampling}
\end{figure}

Much research has centered on the conditions regarding  the   channel and the sampling rate to guarantee   perfect reconstruction of the input signals in MIMO systems.  By contrast, less attention has been paid to  the aspects of   algorithm  design and error analysis.  The problem of FIR reconstruction  for   MIMO sampling of multi-band signals was studied in  \cite{Venka2003filter}. It turns out that for MIMO sampling problems, perfect reconstruction FIR filters do not exist in general.   Kim et al.  \cite{Kim2008vector}
explored the  MIMO sampling  problem by the Riesz basis method and provided  an  unsharp  upper  bound   for the aliasing error in the VSE.  In this paper,  we will focus on the issues of implementation and error analysis for MIMO sampling problems. Unlike the previous works that consider infinite duration signals (see e.g. \cite{Seidner2000vector,shang2007vector}), herein we are interested in  time-limited signals.  A signal defined  on the finite interval $[a,b]$ can be viewed as part of a periodic signal with period $b-a$. There are   methods treating the sampling and reconstruction problems for periodic signals \cite{jacob2002sampling,margolis2008nonuniform,xiao2013sampling,Mohammadi2018sampling}. These works, however, do not involve multi-channel systems. In a   series of papers \cite{Cheng2022signal,cheng2019fft,Cheng2020multi}, the SIMO sampling problem for finite duration signals has been investigated in detail. By the algorithm proposed in \cite{cheng2019fft}, we can perform signal reconstruction from multi-channel samples using only fast Fourier transform (FFT). From the  perspective of FFT interpolation,
the main result  in  \cite{cheng2019fft,Cheng2020multi}  is
a natural extension  to the work of \cite{fraser1989interpolation,selva2015fft} in multi-channel systems. Therefore the approach of SIMO sampling in \cite{cheng2019fft} is also called FFT multi-channel interpolation (FMCI).  Importantly, the FMCI satisfies the interpolation consistency. This implies that resampling the reconstruction via the original sampling operation will produce unchanged measurements. To the authors' knowledge, the consistency of MIMO sampling has not been verified  in some situations (see e.g. \cite{Ma2023nonuniform}).  In this paper, following  \cite{cheng2019fft},  we will provide  an extension to the MIMO sampling expansion (see Fig. \ref{MIMOsampling}), where the input signals of the MIMO channel are finite duration or periodic. Moreover, the  implementation, the consistency and the error analysis for the proposed MIMO sampling and reconstruction method are presented as well.  The   main contributions are highlighted as follows.

\begin{enumerate}
\item The conditions on the MIMO channel and sampling rate that allow perfect reconstruction of   periodic band-limited input signals  are given. Moreover, the reconstruction formula is available   provided that the conditions  for perfect reconstruction are satisfied.
\item  Based on FFT, an efficient and   reliable algorithm is   designed  to compute the value of the reconstructed continuous signal at almost every instant.
\item  We show that   consistency holds for MIMO sampling if the number of output channels is divisible by the number of input channels. Otherwise, the consistency of MIMO sampling is not maintained.
\item The   error  analysis of reconstruction by the proposed method is presented. The expression  of the aliasing   error is given to reveal  how the MIMO system affects the accuracy of reconstruction.    Furthermore, the convergence property of  the proposed method is also verified.
\end{enumerate}

The rest of the paper is organized as follows. Section \ref{S_prelim} introduces   mathematical preliminaries and the problem setting. In Section \ref{S_mimo}, we present the MIMO sampling expansion for finite duration signals.  The conditions for perfect reconstruction, the reconstruction formula as well as its implementation  are investigated. In Section \ref{S_consis_error},  consistency test and   error analysis      are drawn. Section \ref{S_examples} provides examples to illustrate the results. Finally, we conclude the paper in Section \ref{S_conclusion}.

\section{Preliminaries and Problem Statement}\label{S_prelim}

\subsection{Preliminaries}\label{S21}

Throughout the article,     the set of  integers, positive integers, real numbers and complex numbers are denoted  by $\mathbb{Z}$, $\mathbb{Z}^+$, $\mathbb{R}$ and $\mathbb{C}$   respectively, while  the unit circle is represented by $\mathbb{T}:=[0,2\pi)$.
Let $L^p(\T),~1\leq p<\infty$ be the  class of signals $x(t)$ such that
\begin{equation*}
\norm{x}_p:= \left(\frac{1}{2\pi} \int_{\T}|x(t)|^p dt\right)^{\frac{1}{p}}<\infty
\end{equation*}
and $l^p$ be  the totality of sequences $\{a(n)\}_{n\in\mathbb{Z}}$  such that  $\sum_{n\in\mathbb{Z}}\abs{a(n)}^p< \infty$. We  introduce some preparatory knowledge of Fourier series (see e. g. \cite{folland1992fourier,Boehme1975generalized}). For every $x(t)\in L^2(\T)$, we see from Plancherel theorem that its Fourier coefficient sequence $\{a(n)\} $ belongs to $l^2$ and  the Parseval's identity  $\sum_{n\in\mathbb{Z}}\abs{a(n)}^2=\norm{x}_2^2$ holds. It is   known that $L^2(\T)$ is a Hilbert space equipped with  the inner product
\begin{equation*}
(x,y):= \frac{1}{2\pi} \int_{\T} x(t) \overline{y(t)} dt ,\quad \forall x ,y \in  L^2(\T).
\end{equation*}
If $y(t)\in  L^2(\T)$ and its Fourier coefficient sequence is $\{c(n)\} $, then we have the inner product preserving property, namely $$(x,y)=\sum_{n\in\mathbb{Z}}a(n)\overline{c(n)}.$$
Similar to the (aperiodic) convolution defined on $\mathbb{R}$, the cyclic convolution of two
periodic functions $x(t),h(t)$ is defined by
$$(x*h)(t):=\frac{1}{2\pi} \int_{\T}x(s)h(t-s) ds.$$
Suppose that  the Fourier coefficient sequence of $h(t)\in L^2(\T)$ is  $\{b(n)\} $, the convolution theorem states that  the Fourier coefficient sequence of $x*h$ is $\{a(n)b(n)\} $. It is noted that  the convolution theorem is valid in a broader sense, where the  functions may even be generalized functions.

Let $N_1,N_2\in\mathbb{Z}, \mathbf{N}=(N_1,N_2)$,  and $I^{\mathbf{N}}=\{n\in\mathbb{Z}: N_1\leq n \leq N_2\}$, we denote by  $B_{\mathbf{N}}$ the totality of  periodic band-limited functions (trigonometric polynomials) with the following form:
\begin{equation}\label{bandsignal}
x(t)=\sum_{n\in I^{\mathbf{N}}}a(n)e^{\qi nt}  ,~~~I^{\mathbf{N}}=\{n: N_1\leq n \leq N_2\}.
\end{equation}
The bandwidth of $f$ is defined by   the  cardinality of $I^{\mathbf{N}}$, denoted by $\mu(I^{\mathbf{N}})$.  In the paper,   the MIMO sampling expansion for periodic band-limited signals will be established first. Subsequently, we shall discuss what will happen if the proposed sampling formula and algorithm are used to reconstruct non-band-limited signals.

We end this part with a brief review on left inverses of  rectangular  matrices.
If $\mathbf{A}$ is a $m\times n$ ($m< n$) matrix with full column rank, then there are infinitely many left inverses for $\mathbf{A}$ and $\mathbf{A}_{left}^{-1} := \left(\mathbf{A}^{\mathrm{T}} \mathbf{A}\right)^{-1} \mathbf{A}^{\mathrm{T}}$ is a particular  left inverse of $\mathbf{A}$. Furthermore, any left  inverse of $\mathbf{A}$ can be expressed as $\mathbf{A}_{left}^{-1}+\mathbf{B}$, where the rows of $\mathbf{B}$ are vectors in the null space of $\mathbf{A}^{\mathrm{T}}$.

\subsection{Problem Statement}\label{S22}

We consider the MIMO  sampling and reconstruction for periodic signals. The left-hand  side of  Fig. \ref{MIMOsampling} depicts a   MIMO   LTI  system with $R$ inputs $x_1(t), x_2(t),\cdots,x_R(t)$ and $M$ outputs  $y_1(t), y_2(t),\cdots, y_M(t)$. Let $\mathbf{H}(t)$ be   the  impulse response function  matrix of  the MIMO system, then it is a  $M$ by $R$ matrix for relating  the inputs and the outputs.  Denote the $(m,r)$ entry of $\mathbf{H}(t)$ by $h_{mr}(t)$, then the MIMO system can be expressed as
\begin{equation*}
\begin{gathered}
y_1(t)=h_{11} * x_1(t)+h_{12} * x_2(t)+\cdots+h_{1 R} * x_R(t), \\
y_2(t)=h_{21} * x_1(t)+h_{22} * x_2(t)+\cdots+h_{2 R} * x_R(t), \\
\vdots \\
y_M(t)=h_{M 1} * x_1(t)+h_{M 2} * x_2(t)+\cdots+h_{M R} * x_R(t).
\end{gathered}
\end{equation*}
For $1\leq m\leq M$, $1\leq r\leq R$, let
\begin{equation}
x_r(t) = \sum_{n} a_r(n) e^{\qi nt},
\end{equation}
\begin{equation}\label{hFseries}
h_{mr}(t) = \sum_{n} b_{mr}(n) e^{\qi nt},
\end{equation}
\begin{equation}
y_{m}(t) = \sum_{n} c_{m}(n) e^{\qi nt}.
\end{equation}
The convolution theorem indicates that

\begin{equation}\label{relationFourier}
  \begin{pmatrix}
  c_{1}(n)\\
  c_{2}(n)\\
  \vdots\\
  c_{M}(n)
  \end{pmatrix}
  =
  \begin{pmatrix}
  b_{11}(n)&b_{12}(n)&\dots&b_{1R}(n)\\
  b_{21}(n)&b_{22}(n)&\dots&b_{2R}(n)\\
  \vdots&\vdots&\ddots&\vdots\\
  b_{M1}(n)&b_{M2}(n)&\dots&b_{MR}(n)
  \end{pmatrix}
  \begin{pmatrix}
  a_{1}(n)\\
  a_{2}(n)\\
  \vdots\\
  a_{R}(n)
  \end{pmatrix}.
\end{equation}
It is noticed   that the series (\ref{hFseries})
 may not converge point-wise or in the norm of $L^2$.  Nevertheless, $y_m(t)$ is well
defined provided that $\{c_m(n)\}\in l^1$. In this case, $h_{mr} (t)$ may be regarded as a generalized function \cite{Boehme1975generalized}.

Assume that all the  input signals of the MIMO system belong to $B_{\mathbf{N}}$, thus  have the same bandwidth $\mu(I^{\mathbf{N}})$.   The first objective of this paper is to   examine whether the $R$  input signals can be perfectly reconstructed from the samples of the $M$ output signals, as illustrated in the right-hand side of Fig. \ref{MIMOsampling}. The conditions concerning the relationship between \( M \) and \( R \), the impulse response function matrix \(\mathbf{H}(t)\), and the minimum sampling rate need to be specified to achieve perfect reconstruction.
  Furthermore, if these conditions are met, we have to devise suitable reconstruction frameworks.  These issues will be discussed in the next section.

In practice, aliasing  always occurs  when  the original continuous (analog) signals are not band-limited or the continuous signals are sampled  at a rate below the Nyquist rate.  In this work, we are particularly  concerned about what is the difference between the original signal $x_r(t)$ and the reconstructed signal $\hat{x}_r(t)$ if  the conditions for perfect reconstruction are not satisfied. The concept of consistency, initially proposed in \cite{Unser1994nonideal},  has been   employed in the  sampling problem  where there is no restriction of band-limitation on the original signal (see e.g. \cite{xu2017multichannel,Unset1997generalized,Hira2007consistent,poon2014consistent}). It implies that   the reconstructed signal $\hat{x}_r(t)$ is   indistinguishable from  $x_r(t)$ in the sense that they yield  the same measurements under the original sampling operation.
In Section \ref{S_consis_error},   the property of consistency for the proposed sampling and reconstruction framework will be investigated.  It is verified that the  consistency is not guaranteed when $M$ is not divisible by $R$. Another criterion to measure the closeness between $\hat{x}_r(t)$   and  $x_r(t)$ in this paper is the averaged MSE defined by (\ref{averagedMSE}). We will present the expression of this error and discuss how the MIMO LIT system and the sampling rate affect the reconstruction error.

\section{The MIMO Sampling Expansion of Periodic Band-limited Signals}\label{S_mimo}

In this section, we will derive  the MIMO sampling expansion of periodic band-limited signals, and 
provide some examples of MIMO systems allowing perfect reconstruction. Thanks to the distinctive relation between  the samples and the input signals  in frequency domain, an  algorithm based on FFT  is also devised for fast reconstruction.

Since  the input signals are band-limited and the output signals are linear combinations of the filtered versions for the input signals. It follows that the output signals are band-limited and  have the same bandwidth with the input signals. Therefore, to reconstruct $x_r(t)$ ($1\leq r\leq R$) from the samples of $y_m(t)$ ($1\leq m\leq M$) perfectly,   it is necessary to assume that $M\geq R$.  Given  $R$ and $M$, let $m_0 = \lfloor \frac{M}{R}\rfloor$ be the greatest integer not larger than $\frac{M}{R}$ and  denote by $L$   the number of samples in each output channel. To achieve perfect reconstruction of the input signals by FFT, we let
\begin{equation}\label{Ldef}
L = \min\{L'\in \mathbb{Z}^{+}: m_0 L'\geq   \mu(I^{\mathbf{N}})\}.
\end{equation}

\begin{remark}
 Perfect reconstruction can be achieved only if $M L\geq R \mu(I^{\mathbf{N}}) $. Let $\tilde{L} = \min\{L'\in \mathbb{Z}^{+}: M L'\geq R \mu(I^{\mathbf{N}})  \}$. Since $\frac{M}{R}\cdot L\geq \lfloor \frac{M}{R}\rfloor\cdot L=m_0L  $,  then $m_0L\geq \mu(I^{\mathbf{N}})$ implies that $\frac{M}{R}\cdot L\geq \mu(I^{\mathbf{N}})$. It follows that  $L\geq \tilde{L}$, which means that $L$ determined by (\ref{Ldef}) is sufficient  for perfect reconstruction. Moreover, if $M$ is divisible by $R$, we have that $L=\tilde{L}$. In this case, $\frac{L}{2\pi}$ is the minimum of the sampling rate for perfect reconstruction.
\end{remark}

\begin{remark}
If $m_0 L >  \mu(I^{\mathbf{N}})$, there always exists some $I^{\mathbf{N}'}$ such that $I^{\mathbf{N}}\subset I^{\mathbf{N}'}$ and $\mu(I^{\mathbf{N}'})=m_0L$. In this case, $B_{\mathbf{N}}\subset B_{\mathbf{N}'}$ thus every $x(t)\in B_{\mathbf{N}}$ belongs to $B_{\mathbf{N}'}$. Therefore, we  always assume that  $\mu(I^{\mathbf{N}})=m_0L$ if $L$  is determined by (\ref{Ldef}).
\end{remark}

Suppose that $m_0L = \mu(I^{\mathbf{N}})$, we partition $I^{\mathbf{N}}$ into $m_0$  pieces. Let
\begin{equation*}
I_k=\{n: N_1+(k-1)L\leq n\leq N_1+kL-1\},~~J_k= \bigcup_{l=k+1}^{m_0+k}I_l.
\end{equation*}
Then we have $I^{\mathbf{N}}=\bigcup_{k=1}^{m_0} I_k=J_0$  and $\mathbb{Z}=\bigcup_{k\in\mathbb{Z}} I_k $.
For $n\in I_1$, let
\begin{align}
&   \mathbf{a}_r(n)   = \left[a_r(n),a_r({n+L}), \cdots,a_r({n+(m_0-1)L})\right]^{\rm T},\label{fold_a} \\
&  \mathbf{b}_{mr}(n) = \left[b_{mr}(n),b_{mr}({n+L}), \cdots, b_{mr}({n+(m_0-1)L})\right]^{\rm T},\label{fold_b} \\
& \mathbf{e}(n,t)   = \left[e^{\qi nt} ,e^{\qi(n+L)t}, \cdots,e^{\qi(n+(m_0-1)L)t}\right]^{\rm T}.\nonumber
\end{align}
It follows that
\begin{equation}\label{fmatrix}
\begin{split}
x_r(t) & =\sum_{n\in I^{\mathbf{N}}}a_r(n)e^{\qi nt} =\sum_{n=I_1}\sum_{k=0}^{m_0-1}a_r({n+k L}) e^{\qi (n+k L)t} \\
&  = \sum_{n=I_1}\mathbf{e}(n,t)^{\rm T} \mathbf{a}_r(n)  .
\end{split}
\end{equation}
Similar considerations applying to $y_m(t)$, we have
\begin{equation*}
y_m(t)= \sum_{n\in I_1}\mathbf{c}_{m}(n)^{\rm T}   \mathbf{e}(n,t),
\end{equation*}
where
$$\mathbf{c}_{m}(n)  = \left[c(n),c ({n+L}),c({n+2 L}),\cdots,c({n+(m_0 -1)L})\right]^{\rm T}  .$$

The above partitioning technique help us to write the inputs $x_1(t), x_2(t),\cdots,x_R(t)$ and the outputs $y_1(t),y_2(t),\cdots,y_M(t)$ in a relatively compact form and the index set of the summations becomes the fundamental subset of  $I^{\mathbf{N}}$, namely $I_1$. More importantly, we will see that the alternative form of $y_m(t)$ could lead to an expression convenient for computation. Because of the periodicity of $\mathbf{e}(n,t)$, we   have that
\begin{equation*}
\mathbf{e}(n, \tfrac{2\pi p}{L})=e^{\qi n\frac{2\pi p}{L}}[1,1,\cdots,1]^{\rm T}
\end{equation*}
for every $0 \leq p \leq L-1$. It follows that the samples of $y_m(t)$ can be expressed as
\begin{equation}\label{sumformgm}
\begin{split}
y_m(\tfrac{2\pi p}{L}) &  =\sum_{n\in I_1}  \mathbf{c}_{m}(n) ^{\rm T} \mathbf{e}(n,\tfrac{2\pi p}{L}) \\
& = \sum_{n\in I_1}e^{\qi n\frac{2\pi p}{L}} \sum_{k=0}^{m_0 - 1}c_{m}(n+k L)\\
& = \sum_{n\in I_1} d_{m}(n) e^{\qi n\frac{2\pi p}{L}},
\end{split}
\end{equation}
where $ d_{m}(n)=\sum_{k=0}^{m_0-1}c_{m}(n+kL)$ is a folding version of $c_m(n)$. This indicates that  $\{y_m(\frac{2\pi p}{L})\}_p$  and $\{d_m(n)\}_n$  form a pair of discrete Fourier transform (DFT). By the inversion theorem of DFT, we easily obtain 
\begin{equation}\label{dmymp}
d_m(n) = \frac{1}{L}\sum_{p=0}^{L-1}y_m(\tfrac{2\pi p}{L})e^{-\qi n\frac{2\pi p}{L}}, \quad n\in I_1.
\end{equation}

Eq. (\ref{dmymp}) builds a bridge between the   samples and the Fourier coefficients of $y_m(t)$. To reconstruct the input signals from  the   samples of $y_m(t)$,  we need to make use of
the relation between the input signals and the output signals in the frequency domain. By (\ref{relationFourier}) and the definition of $d_m(n)$, we have that
\begin{align*}
d_m(n) & = \sum_{k=0}^{m_0-1}c_{m}(n+kL) \\
      &  = \sum_{k=0}^{m_0-1} \sum_{r=1}^{R} b_{mr}(n+kL) a_{r}(n+kL)\\
      & = \sum_{r=1}^{R} \sum_{k=0}^{m_0-1} b_{mr}(n+kL) a_{r}(n+kL) \\
      & = \widetilde{\mathbf{b}}_{m}(n)^{\rm T}  \widetilde{\mathbf{a}} (n),
\end{align*}
where
\begin{align*}
&\widetilde{\mathbf{b}}_{m}(n)  = \left[\mathbf{b}_{m1}(n)^{\rm T} ,\mathbf{b}_{m2}(n)^{\rm T} ,\cdots,\mathbf{b}_{mR}(n)^{\rm T} \right]^{\rm T},  \\
& \widetilde{\mathbf{a}} (n)     =\left[ \mathbf{a}_1(n)^{\rm T} ,\mathbf{a}_2(n)^{\rm T} ,\cdots, \mathbf{a}_R(n)^{\rm T} \right]^{\rm T},
\end{align*}
and $\mathbf{a}_r(n)$ and $\mathbf{b}_{mr}(n)$ are given by (\ref{fold_a}) and  (\ref{fold_b}) respectively. Let $\mathbf{d}(n) = \left[d_1(n), d_2(n), \cdots, d_M(n) \right]^{\rm T}$ and $
\mathbf{B}(n) =  \left[\widetilde{\mathbf{b}}_{1}(n), \widetilde{\mathbf{b}}_{2}(n) \cdots \widetilde{\mathbf{b}}_{M}(n) \right]^{\rm T}$, we see that
\begin{equation}\label{relation_d_a}
\mathbf{d}(n) = \mathbf{B}(n) \widetilde{\mathbf{a}} (n) .
\end{equation}
Unlike Eq. (\ref{relationFourier}) bridging  $c_m(n)$ and $a_r(n)$,  Eq. (\ref{relation_d_a})  makes a connection between $d_m(n)$ and $a_r(n)$. Based on Eq. (\ref{sumformgm}) and  Eq. (\ref{relation_d_a}), we obtain the following theorem. 

\begin{theorem}\label{MIMOsampling_formula}
Let $x_r(t)\in B_{\mathbf{N}}$, $1\leq r\leq R$ and $y_m(t)$, $1\leq m\leq M$ be the input signals and the output signals of the   MIMO system $\mathbf{H}(t)$. Let  $L$,  $\mathbf{B}(n)$ be given above. If $\mathbf{B}(n)$  is of full column rank for every  $n\in I_1$, then the input signal $x_r(t)$ ($1\leq r\leq R$) can be perfectly reconstructed from the samples of the output signals via
\begin{equation}\label{MIMOexpansion}
x_r(t) = \frac{1}{L} \sum_{m=1}^M \sum_{p=0}^{L-1} y_m(\tfrac{2\pi p}{L})g_{rm}(t-\tfrac{2\pi p}{L}),
\end{equation}
where the $g_{rm}(t)$ is given by (\ref{interp_func}).
\end{theorem}
\begin{proof}
Notice  that $\mathbf{B}(n)$ is a $M$ by $m_0R$ matrix for every  $n\in I_1$. From the definition of $m_0$, we see that $M\geq  m_0R$. If $\mathbf{B}(n)$  is of full column rank for every  $n\in I_1$, then  it is left invertible. Let $\mathbf{Q}(n)$ be a left inverse of $\mathbf{B}(n)$ and 
$$
\mathbf{Q}(n)=
\begin{bmatrix}
q_{1,1}(n) &q_{1,2}(n) &\cdots & q_{1,M}(n)  \\
q_{2,1}(n)& q_{2,2}(n) &\cdots & q_{2,M}(n) \\
\vdots&\vdots&\ddots&\vdots\\
q_{m_0R,1}(n)& q_{m_0R,2}(n) &\cdots & q_{m_0R,M}(n)
\end{bmatrix}.
$$
Based on Eq. (\ref{sumformgm}) and  Eq. (\ref{relation_d_a}), we obtain the relation between  $\mathbf{d}(n)$ and $\widetilde{\mathbf{a}} (n)$. Motivated by the expression (\ref{fmatrix}) as well as   the fact that  $\mathbf{a}_r(n)$ is a part of $\widetilde{\mathbf{a}} (n)$, we
let $\widetilde{\mathbf{e}}_r(n,t) =  \left[ \mathbf{0}_{1\times m_0(r-1)}, \mathbf{e}(n,t)^{\rm T},  \mathbf{0}_{1\times m_0(R-r)} \right]^{\rm T} $, then $x_r(t)$ can be rewritten as
\begin{align*}
x_r(t)& =  \sum_{n=I_1}\mathbf{e}(n,t)^{\rm T} \mathbf{a}_r(n)   \\
      & = \sum_{n=I_1} \widetilde{\mathbf{e}}_r(n,t)^{\rm T} \widetilde{\mathbf{a}} (n). 
\end{align*}
Note that $\mathbf{Q}(n)\mathbf{B}(n)$ is an identity matrix, Plugging it to the above equation, we see that
\begin{align*}
x_r(t) & =  \sum_{n=I_1} \widetilde{\mathbf{e}}_r(n,t)^{\rm T} \mathbf{Q}(n)\mathbf{B}(n)\widetilde{\mathbf{a}} (n) \\
      & = \sum_{n=I_1} \widetilde{\mathbf{e}}_r(n,t)^{\rm T} \mathbf{Q}(n) \mathbf{d}(n).
\end{align*}
Denote $\mathbf{Q}(n)^{\rm T} \widetilde{\mathbf{e}}_r(n,t)$ by $\mathbf{z}_r(n,t)$ and 
let  $ \mathbf{z}_r(n,t) =\left[z_{r1}(n,t),z_{r2}(n,t),\cdots, z_{rM}(n,t) \right]^{\rm T} $, it follows that
\begin{align}
x_r(t)& =  \sum_{n=I_1} \sum_{m=1}^M z_{rm}(n,t) d_m(n) \nonumber \\
      & =  \sum_{n=I_1} \sum_{m=1}^M z_{rm}(n,t) \frac{1}{L}\sum_{p=0}^{L-1} y_m(\tfrac{2\pi p}{L})e^{-\qi n\frac{2\pi p}{L}} \nonumber \\
      & =  \frac{1}{L} \sum_{m=1}^M \sum_{p=0}^{L-1} y_m(\tfrac{2\pi p}{L})\sum_{n=I_1}  z_{rm}(n,t) e^{-\qi n\frac{2\pi p}{L}}. \label{samplingformula1}
\end{align}

It is seen from   (\ref{samplingformula1})  that  $x_r(t)$ can be reconstructed from the samples of the output signals and  $\sum_{n=I_1}  z_{rm}(n,t) e^{-\qi n\frac{2\pi p}{L}}$ acts  as a bridge between continuous signals and discrete samples. Let
$$
\beta_{rm} (n):=
\begin{cases}
 q_{m_0(r-1)+k,m}(n+L-kL),\quad &\text{if} \ \, n \in I_k, \ k=1,2,\cdots,m_0;\\
 0,\quad &\text{if} \ \, n \notin I^{\mathbf{N}};
\end{cases}
$$
and define 
\begin{equation}\label{interp_func}
g_{rm}(t) := \sum_{n\in I^{\mathbf{N}}}\beta_{rm} (n)  e^{\qi n  t}, 
\end{equation}
we claim  that 
$
\sum_{n=I_1}  z_{rm}(n,t) e^{-\qi n\frac{2\pi p}{L}}= g_{rm}(t-\tfrac{2\pi p}{L}).
$
As a consequence, we  arrive    at the reconstruction formula (\ref{MIMOexpansion}).

For any $n\in I_1, 1\leq m\leq M$, it is evident that $\beta_{rm}(n+(k-1)L) = q_{m_0(r-1)+k,m}(n)$. By the definition of $z_{rm}(n,t)$, we get
\begin{align*}
z_{rm}(n,t) & =  \sum_{k=1}^{m_0 } q_{m_0(r-1)+k,m}(n)  e^{\qi (n+(k-1)L)t}  \\
      & =   \sum_{k=1}^{m_0 }  \beta_{rm}(n+(k-1)L) e^{\qi (n+(k-1)L)t}.
\end{align*}
Inserting this expression  into $\sum_{n=I_1}  z_{rm}(n,t) e^{-\qi n\frac{2\pi p}{L}}$ and note that $ e^{\qi n (k-1){2\pi p} }=1$, we obtain
\begin{align*}
\sum_{n=I_1}  z_{rm}(n,t) e^{-\qi n\frac{2\pi p}{L}} & = \sum_{n \in I_1} \sum_{k=1}^{m_0 }  \beta_{rm}(n+(k-1)L) e^{\qi (n+(k-1)L)t} e^{-\qi n\frac{2\pi p}{L}} \\
& = \sum_{k=1}^{m_0 } \sum_{n\in I_1} \beta_{rm}(n+(k-1)L) e^{\qi (n+(k-1)L)(t-\frac{2\pi p}{L})}\\
& = \sum_{n\in I^{\mathbf{N}}}\beta_{rm} (n) e^{\qi n (t- \frac{2\pi p}{L})} \\
& = g_{rm}(t-\tfrac{2\pi p}{L}).
\end{align*}
The proof is complete.
\end{proof}

Having presented the above MIMO sampling expansion, we are in a position to give several  examples of MIMO systems  that allow  perfect reconstruction to show  the appearance of $\mathbf{B}(n)$ and $\mathbf{Q}(n)$.  According to the relationship between $M$ and $R$,   MIMO systems can be divided into three categories.
\begin{itemize}
\item If $M=R$, then  $m_0 = 1$, $L = \mu(I^{\mathbf{N}}) $.  Under this case, $\mathbf{B}(n)$ is a $M\times M$ matrix, and
\begin{equation*}
\mathbf{B}(n) =
\begin{bmatrix}
b_{11}(n) &b_{12}(n) &\cdots & b_{1M}(n)  \\
b_{21}(n)& b_{22}(n) &\cdots & b_{2M}(n) \\
\vdots&\vdots&\ddots&\vdots\\
b_{M1}(n)& b_{M2}(n) &\cdots & b_{MM}(n)
\end{bmatrix}.
\end{equation*}
\item If $M>R$ and $M$ is a multiple of $R$, then  $m_0 = \frac{M}{R}$, $L=\frac{\mu(I^{\mathbf{N}})}{m_0} =\frac{\mu(I^{\mathbf{N}})R}{M}$ and  $\mathbf{B}(n)$ is a $M\times M$ matrix.  For example, when $M=4$ and $R=2$, we have  that $m_0=2$ and
\begin{equation*}
\mathbf{B}(n) =
\begin{bmatrix}
b_{11}(n) &b_{11}(n+L) &b_{12}(n) &b_{12}(n+L)   \\
b_{21}(n)& b_{21}(n+L) & b_{22}(n)  & b_{22}(n+L) \\
b_{31}(n)& b_{31}(n+L) & b_{32}(n)  & b_{32}(n+L) \\
b_{41}(n)& b_{41}(n+L) & b_{42}(n)  & b_{42}(n+L)
\end{bmatrix}.
\end{equation*}
\item  If $M>R$ and $M$ is not a multiple of $R$, then  $\mathbf{B}(n)$ is a $M\times (m_0R)$ matrix.  For example, when $M=3$ and $R=2$,   we have that $m_0=1$ and
\begin{equation*}
\mathbf{B}(n) =
\begin{bmatrix}
b_{11}(n) & b_{12}(n)    \\
b_{21}(n)&   b_{22}(n)   \\
b_{31}(n)&  b_{32}(n)
\end{bmatrix}.
\end{equation*}
\end{itemize}

We provide in Table \ref{example_MIMOsyst}   several examples of MIMO systems that qualify for perfect reconstruction  and the corresponding $\mathbf{Q}(n)$ is also given. It can be seen from  Theorem \ref{MIMOsampling_formula} that $\mathbf{Q}(n)$  plays a crucial role in constructing   interpolating  functions. Theoretically, the analytical expression of the interpolating function can be computed  via (\ref{interp_func}).  But the interpolating function's analytical expression might be very complicated. Nonetheless,  it is not the case that  the  continuous signal  $x_r(t)$ can only be determined by (\ref{MIMOexpansion}).  We are able to compute the values of the reconstructed continuous signals at almost every instant without using the interpolating function $g_{rm}(t)$  directly.

\begin{figure}[ht]
  \centering
  \includegraphics[width=0.9\textwidth]{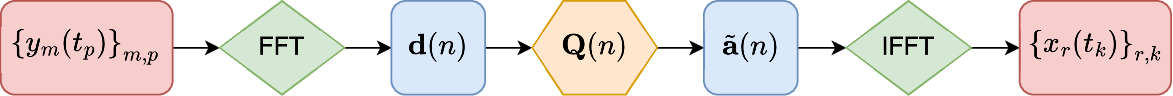}
  \caption{Diagram of FFT-based signal   reconstruction from MIMO samples.}\label{fast_MIMO}
\end{figure}

 Figure \ref{fast_MIMO} offers a glimpse of the foundation for reconstructing the original signals with FFT-based technique. To design an accurate reconstruction algorithm, additional steps such as   modulation and zero padding are required.  These steps help to  unfold the folded frequency, separate signals from various channels, and improve time domain resolution. The proposed Algorithm \ref{alg1} allows us to compute any number of  equally spaced samples of the reconstructed signals.  When $R=1$,  Algorithm \ref{alg1} will reduce to  FFT multi-channel interpolation, and it then turns into FFT interpolation if $M$ is also equal to $1$ and the system is the identity.

{
\begin{table}[htbp]
\centering
\caption{Examples of MIMO systems.}  \label{example_MIMOsyst}
\scriptsize
\vspace{0.2cm}	 
\begin{tabular}{|l|l|l|l|}
\hline
$M,R$  & MIMO system &$ \mathbf{B}(n)$  & $\mathbf{Q}(n)$  \\ \hline
\multirow{4}{*}{\makecell[l]{$M=2$  \\ $R=2$ \\ $m_0=1$}} &  \makecell[l]{ ~\\
 $y_1(t) = x_1(t) + x'_2(t)$\\
 $y_2(t) = x'_1(t) + x_2(t)$\\ ~} &  \makecell[l]{ $\begin{bmatrix}
 1 & \qi n \\
 \qi n & 1 \\
\end{bmatrix}$ } & \makecell[l]{$\begin{bmatrix}
 \frac{1}{n^2+1} & -\frac{\qi n}{n^2+1} \\
 -\frac{\qi n}{n^2+1} & \frac{1}{n^2+1} \\
\end{bmatrix}$} \\ \cline{2-4} 
                & \makecell[l]{~\\$y_1(t) = x_1(t) + x_2(t+\alpha)$\\
 $y_2(t) = x_1(t+\alpha) + 2x_2(t)$\\ ~}   &  \makecell[l]{$\begin{bmatrix}
 1 & e^{\qi n \alpha} \\
e^{\qi n \alpha}  & 2 \\
\end{bmatrix}$}  & \makecell[l]{$ \begin{bmatrix}
 \frac{2}{2-e^{2 \qi \alpha  n}} & -\frac{e^{\qi \alpha  n}}{2-e^{2 \qi \alpha  n}} \\
 -\frac{e^{\qi  n \alpha }}{2-e^{2 \qi n \alpha }} & \frac{1}{2-e^{2 \qi   n\alpha}} \\
\end{bmatrix}$}  \\ \cline{2-4} 
                   &  \makecell[l]{~\\$y_1(t) = x_1(t) + x'_2(t)$\\
 $y_2(t) = x_1(t) +  x'_1(t) + x_2(t)$\\~} &  \makecell[l]{$\begin{bmatrix}
 1 & \qi n \\
 1+\qi n & 1 \\
\end{bmatrix} $} &  \makecell[l]{$\begin{bmatrix}
 \frac{1}{n^2 -\qi n+1} &  \frac{-\qi n}{n^2-\qi n+1} \\
 \frac{-1-\qi n}{n^2-\qi n+1} & \frac{1}{n^2-\qi n+1} \\
\end{bmatrix}$} \\ \cline{2-4} 
                   &  \makecell[l]{~\\$y_1(t) = x_1(t) +x_2(t+\alpha)$\\
 $y_2(t) =  x'_1(t) + x_2(t)$\\~} & \makecell[l]{$\begin{bmatrix}
 1 & e^{\qi n \alpha}  \\
 \qi n & 1 \\
\end{bmatrix}$}  & \makecell[l]{$ \begin{bmatrix}
 \frac{1}{1-\qi n e^{\qi \alpha  n}} & \frac{-e^{\qi \alpha  n}}{1-\qi n e^{\qi \alpha  n}} \\
 \frac{-\qi n}{1-\qi n e^{\qi \alpha  n}} & \frac{1}{1-\qi n e^{\qi \alpha  n}} \\
\end{bmatrix}$}   \\ \hline
\multirow{2}{*}{\makecell[l]{$M=3$  \\ $R=2$ \\ $m_0=1$}} & \makecell[l]{$y_1(t) = x_1(t) +x_2(t+\alpha)$\\$y_2(t) =  x_1(t+\alpha) + x_2(t)$\\$y_3(t) =  2x_1(t) + x_2(t)$ } &   \makecell[l]{$\begin{bmatrix}
1 & e^{\qi n \alpha}   \\
e^{\qi n \alpha} &   1  \\
 2 &  1
\end{bmatrix}$}  &  \makecell[l]{~\\  $q_{11}(n)=(2-2 e^{\qi \alpha  n}-e^{2 \qi \alpha  n})/\sigma_1$ \\   $q_{12}(n)=(e^{3 \qi \alpha  n}-2)/\sigma_1$\\   $q_{13}(n)=2 \left(1-e^{\qi \alpha  n}+e^{2 \qi \alpha  n} \right)/\sigma_1$\\   $q_{21}(n)=(3 e^{\qi \alpha  n}+e^{3 \qi \alpha  n}-2)/\sigma_1$ \\   $q_{22}(n)=(5-2 e^{\qi \alpha  n}-e^{2 \qi \alpha  n})/\sigma_1 $\\   $q_{23}(n)=(1-4 e^{\qi \alpha  n}+e^{2 \qi \alpha  n})/\sigma_1  $ \\   $\sigma_1={6-8 e^{\qi \alpha  n}+3 e^{2 \qi \alpha  n}+e^{4 \qi \alpha  n}}$ \\~} \\ \cline{2-4} 
                   &\makecell[l]{ ~\\$y_1(t) = x_1(t) +x_2(t)$\\
$y_2(t) =  x_1(t) + x'_2(t)$ \\
$y_3(t) =  x'_1(t) + x_2(t)$\\~}   & $\begin{bmatrix}
1 & 1   \\
1 &   \qi n  \\
\qi n &  1
\end{bmatrix}$  & $\begin{bmatrix}
 \frac{1}{2 \qi n+3-n^2} & \frac{2 \qi-n}{\qi n^2+n+3 \qi-n^3} & \frac{\qi n^2+n-\qi}{\qi n^2+n+3 \qi-n^3} \\
 \frac{1}{2 \qi n+3-n^2} &\frac{\qi n^2+n-\qi}{\qi n^2+n+3 \qi-n^3} & \frac{2 \qi-n}{\qi n^2+n+3 \qi-n^3} \\
\end{bmatrix}$  \\ \hline
\multirow{2}{*}{\makecell{$M=4$  \\ $R=2$ \\ $m_0=2$}} & \makecell[l]{~\\
$y_1(t) = 2x_1(t) +x_2(t)$\\
$y_2(t) =  x_1(t) + x'_2(t)$ \\
$y_3(t) =  x'_1(t) + x_2(t)$ \\
$ y_4(t) =  x'_1(t) + x'_2(t)$\\~}  &  $\begin{bmatrix}
2 & 2 & 1 & 1  \\
1 & 1 & \qi n  &\qi(n+L) \\
\qi n & \qi(n+L) &  1  & 1  \\
\qi n & \qi(n+L) & \qi n  &\qi(n+L)
\end{bmatrix}$   &  $\begin{bmatrix}
 \frac{L+n+\qi}{L} & \frac{L+n+2 \qi}{-L} & \frac{L+n+\qi}{-L} & \frac{L+n+2 \qi}{L} \\
 -\frac{n+\qi}{L} & \frac{n+2 \qi}{L} & \frac{n+\qi}{L} & -\frac{n+2 \qi}{L} \\
 \frac{L+n+\qi}{-L} & \frac{2 (L+n+\qi)}{L} & \frac{2 L+2 n+\qi}{L} &  \frac{2 L+2 n+\qi}{-L} \\
 \frac{n+\qi}{L} & -\frac{2 (n+\qi)}{L} & -\frac{2 n+\qi}{L} & \frac{2 n+\qi}{L} \\
\end{bmatrix}$  \\ \cline{2-4} 
                   &  \makecell[l]{ $y_1(t) = 2x_1(t) +x_2(t)$\\
$y_2(t) =  x_1(t+\alpha) + x_2(t) $\\
$y_3(t) =  x_1(t) + x_2(t+\alpha)$ \\
$ y_4(t) =  x_1(t+\alpha) + x_2(t+\alpha)$}  & \makecell[l]{~\\ [5pt] $\begin{bmatrix}
 2 & 2 & 1 & 1 \\
\sigma_2^{-1} &\sigma_3\sigma_2^{-1} & 1 & 1 \\
 1 & 1 & \sigma_2^{-1} &\sigma_3\sigma_2^{-1} \\
\sigma_2^{-1} & \sigma_3\sigma_2^{-1}& \sigma_2^{-1} &\sigma_3\sigma_2^{-1} \\
\end{bmatrix} $\\  ~ }&  \makecell[l]{
$\begin{bmatrix}
 \frac{\sigma_2-\sigma_3}{1-\sigma_3} &  \frac{\sigma_3-\sigma_2}{1-\sigma_3} & \frac{\sigma_3-2\sigma_2}{1-\sigma_3}& \frac{2\sigma_2-\sigma_3}{1-\sigma_3}  \\[2pt]
 \frac{1-\sigma_2}{1-\sigma_3} &  \frac{\sigma_2-1}{1-\sigma_3}& \frac{2\sigma_2-1}{1-\sigma_3}& \frac{1-2\sigma_2}{1-\sigma_3} \\[2pt]
 \frac{\sigma_3-\sigma_2}{1-\sigma_3} &   \frac{\sigma_2-2\sigma_3}{1-\sigma_3} &  \frac{2\sigma_2-2\sigma_3}{1-\sigma_3} & \frac{2\sigma_3-\sigma_2}{1-\sigma_3} \\[2pt]
 \frac{\sigma_2-1}{1-\sigma_3} &\frac{2-\sigma_2}{1-\sigma_3} &\frac{2-2\sigma_2}{1-\sigma_3} &\frac{\sigma_2-2}{1-\sigma_3} 
\end{bmatrix}$\\ 
$\sigma_2= e^{-\qi \alpha  n}$, $\sigma_3= e^{\qi \alpha  L}$
}  \\ \hline
\end{tabular}
\end{table}
}

\IncMargin{1em} 
\begin{algorithm}[htb]

    \SetAlgoNoLine 
    \SetKwInOut{Input}{\textbf{Input}}\SetKwInOut{Output}{\textbf{Output}} 

    \Input{
    \\
    1. The MIMO samples $\mathbf{Y}=[\mathbf{y}_1,\mathbf{y}_2,\dots,\mathbf{y}_M]$, where $\mathbf{y}_m$ is the    $L$ samples of   \\~~~   $y_m$,  i.e., $\mathbf{y}_m=[y_m(t_0),y_m(t_1),\cdots,y_m(t_{L-1})]^{\rm T}$ with $t_p=\frac{2\pi p}{L}$;\\
    2.  The location of lower bound for the frequency band: $N_1$ \;\\
     3.   The number of  function values of  $\mathcal{T}_\mathbf{N} x_r(t)$:  $N^o_{r}$.\ \\}
    \Output{
        \\
         1.  The vector $\mathbf{x}^o_r$ consisting   of $N^o_{r}$ function values of $\mathcal{T}_\mathbf{N} x_r(t)$. \medskip \\}
    \BlankLine

    Multiply $k$-th row of $\mathbf{Y}$ by $\frac{1}{L}e ^{   \frac{-2\pi\qi N_1 (k-1)}{L} }$, obtain ${\mathbf{Y}_e}$ for $1\leq k\leq L$\;
    Take FFT of ${\mathbf{Y}_e}$ (for each column), obtain $\widehat{\mathbf{Y}}_e$\;
    Compute $m_0R\times L$ matrix $\widetilde{\mathbf{A}}$, where $\widetilde{\mathbf{A}}[:,k]=\mathbf{Q}(N_1+k-1) \widehat{\mathbf{Y}}_e[k,:]^{\rm T}$ for $1\leq k\leq L$\;
    Flatten $\widetilde{\mathbf{A}}$  w.r.t. row, reshape it as the vector named $\mathbf{a}_{all}$ with length $m_0RL$\;
     Let   $\mathbf{f}_r = \mathbf{a}_{all}[m_0L(r-1)+1:m_0Lr]$ for $1\leq r\leq R$\;
    Zero padding: add $N^o_r- m_0L$ zeros at the end of $\mathbf{f}_r$, obtain $\mathbf{f}_r^z$ for $1\leq r\leq R$\;
    Take IFFT of $\mathbf{f}_r^z$  and multiply its $k$-th element    by ${N_r^o}  e^{  \frac{2\pi\qi N_1 (k-1)}{N_r^o} } $,   get  $\mathbf{x}^o_r$  ($1\leq r\leq R$).
    \caption{FFT-based signal reconstruction algorithm for MIMO samples.}\label{alg1}
\end{algorithm}
\DecMargin{1em}

\section{Consistency Test and Error analysis}\label{S_consis_error}

The MIMO sampling expansion (\ref{MIMOexpansion}) and its FFT-based implementation  are provided to reconstruct periodic band-limited signals from MIMO samples.  In practice, however, the conditions for perfect reconstruction are generally unsatisfied. For example, if the  signal to be reconstructed is  non-band-limited,  or the sampling rate is too low, the aliasing error is introduced. Define
\begin{equation}\label{reconst-formula}
\hat{x}_r(t) =[\mathcal{T}_{\mathbf{N}} x_r](t) := \frac{1}{L} \sum_{m=1}^M \sum_{p=0}^{L-1} y_m(\tfrac{2\pi p}{L}) g_{rm}(t-\tfrac{2\pi p}{L}),\quad  1\leq r \leq R, 
\end{equation}
we will investigate the difference between the original signal $x_r(t)$ and the reconstructed signal $\hat{x}_r(t)$ in this section.

\subsection{Consistency Test}

Let
\begin{equation}
\hat{y}_m(t) = h_{m1} * \hat{x}_1(t)+h_{m2} * \hat{x}_2(t)+\cdots+h_{m R} * \hat{x}_R(t)  , \quad  1\leq m \leq M.
\end{equation}
 If the following equality holds
\begin{equation}
\hat{y}_m(\tfrac{2\pi p}{L}) =  {y}_m(\tfrac{2\pi p}{L}), \quad  0\leq p\leq L-1,~  1\leq m \leq M,
\end{equation}
 then we say the MIMO sampling expansion satisfies the  extended interpolation consistency.

When $M$ is divisible by $R$, we have $M=m_0R$. By the definition of $\beta$, we see that
\begin{equation*}
\beta_{rj}(n+sL)=q_{m_0(r-1)+s+1,j}(n),
\end{equation*}
i.e., the entry of $\mathbf{Q}(n)$ in row $m_0(r-1)+s+1$ and column $j$. On the other hand,
from the construction of  $\mathbf{B}(n)$, $b_{m,r}(n+sL)$ is the  entry of $\mathbf{B}(n)$ in
row $m$ and column $m_0(r-1)+s+1$. It follows that
\begin{equation*}
\sum_{r=1}^R\sum_{s=0}^{m_0-1}b_{m,r}(n+sL)\beta_{rj}(n+sL)
\end{equation*}
is the entry of $\mathbf{B}(n)\mathbf{Q}(n)$ in row $m$ and column $j$. Note that  $\mathbf{Q}(n)$ is a left inverse of $\mathbf{B}(n)$. When $M=m_0R$, both $\mathbf{B}(n)$ and  $\mathbf{Q}(n)$ are square matrices. Therefore  $\mathbf{Q}(n)$ is exactly the (two-sided) inverse of $\mathbf{B}(n)$. That is $\mathbf{B}(n)\mathbf{Q}(n)=\mathbf{Q}(n)\mathbf{B}(n)=\mathbf{I}$. It follows that
\begin{equation*}
\sum_{r=1}^R\sum_{s=0}^{m_0-1}b_{m,r}(n+sL)\beta_{rj}(n+sL) =\delta(m-j).
\end{equation*}

Note that the Fourier coefficient of $\hat{x}_r(t)$ is $\frac{1}{L} \sum_{j=1}^M \sum_{k=0}^{L-1} y_j(\tfrac{2\pi k}{L})\beta_{rj}(n)e^{-\qi n \frac{2\pi k}{L}}$, then
\begin{align*}
\hat{y}_m(\tfrac{2\pi p}{L}) &= \sum_{r=1}^R \hat{x}_r *h_{mr} (\tfrac{2\pi p}{L}) \\
& =  \sum_{r=1}^R \sum_{n\in I^{\mathbf{N}}}\frac{1}{L} \sum_{j=1}^M \sum_{k=0}^{L-1} y_j(\tfrac{2\pi k}{L})\beta_{rj}(n)e^{-\qi n \frac{2\pi k}{L}} b_{mr}(n)e^{\qi n \frac{2\pi p}{L}} \\
& = \sum_{r=1}^R \sum_{n\in I_1} \sum_{s=0}^{m_0-1}\frac{1}{L} \sum_{j=1}^M \sum_{k=0}^{L-1} y_j(\tfrac{2\pi k}{L})\beta_{rj}(n+sL)e^{-\qi (n+sL) \frac{2\pi k}{L}} b_{mr}(n+sL)e^{\qi (n+sL) \frac{2\pi p}{L}} \\
&= \sum_{r=1}^R \sum_{n\in I_1} \sum_{s=0}^{m_0-1}\frac{1}{L} \sum_{j=1}^M \sum_{k=0}^{L-1} y_j(\tfrac{2\pi k}{L})\beta_{rj}(n+sL) b_{mr}(n+sL) e^{\qi n \frac{2\pi (p-k)}{L}}  \\
&= \sum_{r=1}^R \sum_{n\in I_1} \sum_{s=0}^{m_0-1}\frac{1}{L} \sum_{j=1}^M \sum_{k=0}^{L-1} y_j(\tfrac{2\pi k}{L})\beta_{rj}(n+sL) b_{mr}(n+sL) e^{\qi n \frac{2\pi (p-k)}{L}}  \\
\end{align*}
Taking the summation over $r,s,n$ in turn, and
using the equality $\frac{1}{L}\sum_{n\in I_1} e^{\qi n \frac{2\pi (p-k)}{L}}= \delta(p-k)$ for $0\leq p,k\leq L-1$, we have that
\begin{align*}
\hat{y}_m(\tfrac{2\pi p}{L}) = \sum_{j=1}^M \sum_{k=0}^{L-1} y_j(\tfrac{2\pi k}{L})\delta(m-j)\delta(p-k)={y}_m(\tfrac{2\pi p}{L}),
\end{align*}
which completes the proof of the extended interpolation consistency.

When $M$ is  not divisible by $R$, we have $m_0R<M$. It implies that  $\mathbf{Q}(n)$ is not a  right inverse of $\mathbf{B}(n)$. Thus   we cannot be sure if the  extended interpolation consistency still hold for this case. We will give an example to show  the extended interpolation consistency is not valid for $M=3,R=2$ in Section \ref{S_examples} (see Figure \ref{consistency_test_2}).

Although the   extended interpolation consistency does not hold for the case  $m_0R<M$,
the ``inconsistency"
\begin{equation*}
\sum_{m=1}^M\sum_{p=0}^{L-1}\abs{\hat{y}(\tfrac{2\pi p}{L})-{y}(\tfrac{2\pi p}{L})}^2
\end{equation*}
will decrease  as the number of samples increases. This is because  regardless of whether the   extended interpolation consistency  holds or not, $\hat{x}_r(t)$ would tend to ${x}_r(t)$  as the number of samples goes to infinity. The next part is devoted to   the error analysis of the  reconstruction  by the proposed MIMO sampling expansion.

\subsection{Theoretical Error analysis}

If some of  the original input signals are not band-limited, the formula (\ref{reconst-formula}) still provide an approximation for $x_r(t)$. Let $x_{r,\tau}(t):=x_r(t-\tau)$ be the translated signal of $x_r(t)$, we see that
\begin{equation}
[\mathcal{T}_{\mathbf{N}} x_{r,\tau}](t) = \frac{1}{L} \sum_{m=1}^M \sum_{p=0}^{L-1} y_m(\tfrac{2\pi p}{L}-\tau) g_{rm}(t-\tfrac{2\pi p}{L}),\quad  1\leq r \leq R .
\end{equation}
It can be verified that  $\mathcal{T}_{\mathbf{N}}$ is not   shift-invariant,  which means that
$[\mathcal{T}_{\mathbf{N}} x_{r,\tau}](t) \neq [\mathcal{T}_{\mathbf{N}} x_{r}](t-\tau)$.  We consider the following averaged MSE:
\begin{equation}\label{averagedMSE}
\epsilon(x_r,\mathbf{N}) :=\sqrt{\frac{L}{2\pi}\int_0^{\frac{2\pi}{L}}\frac{1}{2\pi}\int_{\mathbb{T}}\abs{ x_{r,\tau}(t) -[\mathcal{T}_{\mathbf{N}} x_{r,\tau}](t)}^2 dt d\tau}.
\end{equation}
It is not difficult to observe that as  $L$ approaches infinity,  $\epsilon(x_r,\mathbf{N})$ will tend towards the actual MSE:
\begin{equation}\label{actualMSE}
\xi(x_r,\mathbf{N}) :=\sqrt{ \frac{1}{2\pi}\int_{\mathbb{T}}\abs{ x_{r}(t) -[\mathcal{T}_{\mathbf{N}} x_{r}](t)}^2 dt  }.
\end{equation}
Therefore, when $L$  is large, $\epsilon(x_r,\mathbf{N})$ is a good approximation of  $\xi(x_r,\mathbf{N})$.

Squaring both sides of (\ref{averagedMSE}) yields
\begin{align*}
\epsilon^2(x_r,\mathbf{N}) = & \frac{L}{2\pi}\int_0^{\frac{2\pi}{L}}\frac{1}{2\pi}\int_{\mathbb{T}}\abs{ x_{r,\tau}(t) -[\mathcal{T}_{\mathbf{N}} x_{r,\tau}](t)}^2 dt d\tau \\
 =  & \frac{L}{2\pi}\int_0^{\frac{2\pi}{L}}\frac{1}{2\pi}\int_{\mathbb{T}}\abs{ x_{r,\tau}(t)}^2 dt d\tau \\
& + \frac{L}{2\pi}\int_0^{\frac{2\pi}{L}}\frac{1}{2\pi}\int_{\mathbb{T}}[\mathcal{T}_{\mathbf{N}} x_{r,\tau}](t) \overline{x_{r,\tau}(t)}  dt d\tau   \\
& + \frac{L}{2\pi}\int_0^{\frac{2\pi}{L}}\frac{1}{2\pi}\int_{\mathbb{T}} \overline{[\mathcal{T}_{\mathbf{N}} x_{r,\tau}](t)} x_{r,\tau}(t)  dt d\tau  \\
& +\frac{L}{2\pi}\int_0^{\frac{2\pi}{L}} \frac{1}{2\pi}\int_{\mathbb{T}} \abs{[\mathcal{T}_{\mathbf{N}} x_{r,\tau}](t) }^2dt d\tau.
\end{align*}
We will calculate the above four items separately.

Note that the Fourier coefficient  of $x_{r,\tau}(t)$ is $a_r(n)e^{-\qi n \tau}$, it is easy to see that
$$
 \frac{L}{2\pi}\int_0^{\frac{2\pi}{L}}\frac{1}{2\pi}\int_{\mathbb{T}}\abs{ x_{r,\tau}(t)}^2 dt d\tau= \sum_{n\in \mathbb{Z}} \abs{a_r(n)}^2.
$$
Let
$$
v_m(n,\tau)  =\sum_{k\in\mathbb{Z}}\sum_{j=1}^R a_j(n+kL) b_{mj}(n+kL) e^{-\qi k L\tau},
$$
it follows that the Fourier coefficient  of $[\mathcal{T}_{\mathbf{N}} x_{r,\tau}](t)$ is $e^{-\qi n \tau} \sum_{m=1}^M v_m(n,\tau) \beta_{rm}(n)$.
By direct computation, we have
\begin{align*}
 &\frac{L}{2\pi}\int_{0}^{\frac{2\pi}{L}}v_m(n,\tau) d\tau \\
=& \frac{L}{2\pi}\int_{0}^{\frac{2\pi}{L}}\sum_{k\in\mathbb{Z}}\sum_{j=1}^R a_j(n+kL) b_{mj}(n+kL) e^{-\qi k L\tau} d\tau \\
=& \sum_{k\in\mathbb{Z}}\sum_{j=1}^R a_j(n+kL) b_{mj}(n+kL)   \frac{L}{2\pi}\int_{0}^{\frac{2\pi}{L}}e^{-\qi k L\tau} d\tau \\
=& \sum_{k\in\mathbb{Z}}\sum_{j=1}^R a_j(n+kL) b_{mj}(n+kL)   \delta(k) \\
=& \sum_{j=1}^R a_j(n) b_{mj}(n).
\end{align*}
For every $n\in I^{\mathbf{N}}$,
\begin{equation}\label{prod_delta}
\sum_{m=1}^M \beta_{rm}(n) b_{mj}(n) = \delta(r-j).
\end{equation}
By Plancherel theorem we have that
\begin{align*}
 & \frac{L}{2\pi}\int_0^{\frac{2\pi}{L}}\frac{1}{2\pi}\int_{\mathbb{T}}[\mathcal{T}_{\mathbf{N}} x_{r,\tau}](t) \overline{x_{r,\tau}(t)}  dt d\tau  \\
 =& \frac{L}{2\pi}\int_0^{\frac{2\pi}{L}} \sum_{n\in I^{\mathbf{N}}} \overline{a_r(n)} \sum_{m=1}^M v_m(n,\tau)\beta_{rm}(n) d\tau \\
 =&  \sum_{n\in I^{\mathbf{N}}} \overline{a_r(n)} \sum_{m=1}^M \beta_{rm}(n)\frac{L}{2\pi}\int_0^{\frac{2\pi}{L}}v_m(n,\tau)  d\tau \\
 =& \sum_{n\in I^{\mathbf{N}}} \overline{a_r(n)} \sum_{m=1}^M \beta_{rm}(n)\sum_{j=1}^R a_j(n) b_{mj}(n) \\
 =& \sum_{n\in I^{\mathbf{N}}} \overline{a_r(n)} \sum_{j=1}^R a_j(n)\sum_{m=1}^M \beta_{rm}(n)b_{mj}(n) \\
 =& \sum_{n\in I^{\mathbf{N}}} \abs{a_r(n)}^2.
\end{align*}
Taking the conjugate of $\frac{L}{2\pi}\int_0^{\frac{2\pi}{L}}\frac{1}{2\pi}\int_{\mathbb{T}}[\mathcal{T}_{\mathbf{N}} x_{r,\tau}](t) \overline{x_{r,\tau}(t)}  dt d\tau$, it gives
$$
\frac{L}{2\pi}\int_0^{\frac{2\pi}{L}}\frac{1}{2\pi}\int_{\mathbb{T}} \overline{[\mathcal{T}_{\mathbf{N}} x_{r,\tau}](t)} x_{r,\tau}(t)  dt d\tau  = \sum_{n\in I^{\mathbf{N}}} \abs{a_r(n)}^2.
$$

Now we calculate the last term of $\epsilon^2(x_r,\mathbf{N})$. Applying Plancherel theorem again,  we  obtain
\begin{align}
 &   \frac{L}{2\pi}\int_0^{\frac{2\pi}{L}} \frac{1}{2\pi}\int_{\mathbb{T}} \abs{[\mathcal{T}_{\mathbf{N}} x_{r,\tau}](t) }^2dt d\tau \nonumber \\
= &   \frac{L}{2\pi}\int_0^{\frac{2\pi}{L}} \sum_{n\in I^{\mathbf{N}}} \sum_{m=1}^M v_m(n,\tau) \beta_{rm}(n) \sum_{l=1}^M \overline{v_l(n,\tau)} \overline{\beta_{rl}(n)} d\tau \nonumber \\
=& \sum_{n\in I^{\mathbf{N}}} \sum_{m=1}^M \sum_{l=1}^M  \beta_{rm}(n) \overline{\beta_{rl}(n)} \frac{L}{2\pi}\int_0^{\frac{2\pi}{L}} v_m(n,\tau) \overline{v_l(n,\tau)} d\tau \nonumber \\
=& \sum_{k\in \mathbb{Z}}\sum_{n\in I^{\mathbf{N}}}  \sum_{m=1}^M \sum_{l=1}^M \sum_{j=1}^R \sum_{i=1}^R a_j(n+kL)b_{mj}(n+kL) \beta_{rm}(n) \overline{a_i(n+kL)}\overline{b_{li}(n+kL)}\overline{\beta_{rl}(n)} \nonumber \\
=&  \sum_{k\in \mathbb{Z}}\sum_{n\in I^{\mathbf{N}}}  \abs{\sum_{m=1}^M \sum_{j=1}^Ra_j(n+kL)b_{mj}(n+kL) \beta_{rm}(n)}^2 \nonumber \\
=&  \sum_{k\in \mathbb{Z}} \sum_{n\in J_k} \abs{\sum_{m=1}^M \sum_{j=1}^Ra_j(n)b_{mj}(n) \beta_{rm}(n-kL)}^2. \label{norm_hat_x_r}
\end{align}
In the above derivation, we used the following equality:
\begin{align*}
 &   \frac{L}{2\pi}\int_0^{\frac{2\pi}{L}} v_m(n,\tau) \overline{v_l(n,\tau)} d\tau \\
= &  \sum_{k\in \mathbb{Z}} \sum_{j=1}^R a_j(n+kL)b_{mj}(n+kL)\sum_{i=1}^R \overline{a_i(n+kL)}\overline{b_{li}(n+kL)}.
\end{align*}

We divide (\ref{norm_hat_x_r}) into three parts:
\begin{align}
 &\frac{L}{2\pi}\int_0^{\frac{2\pi}{L}} \frac{1}{2\pi}\int_{\mathbb{T}} \abs{[\mathcal{T}_{\mathbf{N}} x_{r,\tau}](t) }^2dt d\tau \nonumber \\ =& \sum_{n\in J_0} \abs{\sum_{m=1}^M \sum_{j=1}^R a_j(n)b_{mj}(n) \beta_{rm}(n)}^2 \nonumber \\
&+   \sum_{k\neq 0} \sum_{n\in J_k\cap I^{\mathbf{N}}} \abs{\sum_{m=1}^M \sum_{j=1}^R a_j(n)b_{mj}(n) \beta_{rm}(n-kL)}^2 \nonumber \\
&+  \sum_{k\neq 0} \sum_{n\in J_k\setminus I^{\mathbf{N}}} \abs{\sum_{m=1}^M \sum_{j=1}^R a_j(n)b_{mj}(n) \beta_{rm}(n-kL)}^2 . \nonumber
\end{align}
Applying (\ref{prod_delta}) and note that
$$
\sum_{m=1}^M b_{mj}(n_1) \beta_{rm}(n_2)=0
$$
for every $n_1,n_2\in I^{\mathbf{N}}$ satisfying $\frac{n_1-n_2}{L}\in \mathbb{Z}\setminus \{0\}$, we see that
$$
  \sum_{n\in J_0} \abs{\sum_{m=1}^M \sum_{j=1}^R a_j(n)b_{mj}(n) \beta_{rm}(n)}^2=\sum_{n\in  I^{\mathbf{N}}} \abs{a_r(n)}^2 
$$
and
$$
  \sum_{k\neq 0} \sum_{n\in J_k\cap I^{\mathbf{N}}} \abs{\sum_{m=1}^M \sum_{j=1}^R a_j(n)b_{mj}(n) \beta_{rm}(n-kL)}^2= 0 .
$$
It follows that 
\begin{equation}\label{exact_alias_error}
\epsilon^2(x_r,\mathbf{N})  =  \sum_{n\notin  I^{\mathbf{N}}} \abs{a_r(n)}^2 + \sum_{k\neq 0} \sum_{n\in J_k\setminus I^{\mathbf{N}}} \abs{\sum_{m=1}^M \sum_{j=1}^R a_j(n)b_{mj}(n) \beta_{rm}(n-kL)}^2. 
\end{equation}
By now, the  averaged MSE of the reconstruction by (\ref{reconst-formula}) has been given.
In  addition to this, we also provide some upper bounds for $\epsilon^2(x_r,\mathbf{N})$.

Let
$$
\beta_{max,\mathbf{N} } :=\max_{n\in I^\mathbf{N}} \{ \abs{\beta_{rm}(n)}: 1\leq r\leq R,1\leq m\leq M\}.
$$
Note that for every $n\in J_k\setminus I^{\mathbf{N}}$, we have $n-kL\in I^\mathbf{N}$, it follows that
$$
\abs{\beta_{rm}(n-kL)} \leq \beta_{max,\mathbf{N} },\quad  n\in J_k\setminus I^{\mathbf{N}}.
$$
Therefore
\begin{align}
\epsilon^2(x_r,\mathbf{N})  = & \sum_{n\notin  I^{\mathbf{N}}} \abs{a_r(n)}^2 + \sum_{k\neq 0} \sum_{n\in J_k\setminus I^{\mathbf{N}}} \abs{\sum_{m=1}^M \sum_{j=1}^R a_j(n)b_{mj}(n) \beta_{rm}(n-kL)}^2  \nonumber \\
\leq & \sum_{n\notin  I^{\mathbf{N}}} \abs{a_r(n)}^2 + \sum_{k\neq 0} \sum_{n\in J_k\setminus I^{\mathbf{N}}}  \sum_{m=1}^M \sum_{j=1}^R \abs{a_j(n)b_{mj}(n) \beta_{rm}(n-kL)}^2 \nonumber  \\
\leq &  \sum_{n\notin  I^{\mathbf{N}}} \abs{a_r(n)}^2 +\beta_{max,\mathbf{N}}^2\sum_{k\neq 0} \sum_{n\in J_k\setminus I^{\mathbf{N}}}  \sum_{m=1}^M \sum_{j=1}^R \abs{a_j(n)b_{mj}(n)}^2 \nonumber \\
\leq & \sum_{n\notin  I^{\mathbf{N}}} \abs{a_r(n)}^2 + m_0\beta_{max,\mathbf{N}}^2 \sum_{m=1}^M \sum_{j=1}^R \sum_{n\notin  I^{\mathbf{N}}} \abs{a_j(n)b_{mj}(n)}^2. \label{upperbound}
 \end{align}
The last inequality is due to the following facts.
\begin{enumerate}
\item [(1)] If $n \in  I^{\mathbf{N}}$, it won't appear in the summation $\sum_{k\neq 0} \sum_{n\in J_k\setminus I^{\mathbf{N}}}$.
\item [(2)] If $n\notin  I^{\mathbf{N}}$, it   reappears not more than $m_0$ times in the   summation $\sum_{k\neq 0} \sum_{n\in J_k\setminus I^{\mathbf{N}}}$.
\end{enumerate}

Eq. (\ref{upperbound}) gives a  loose upper bound of the reconstruction error. Nonetheless, we can see from (\ref{upperbound})  that   if $\beta_{max,\mathbf{N}}$ is bounded over $\mathbf{N}$ then $\epsilon^2(x_r,\mathbf{N})\to 0$ as $\mathbf{N}\to(-\infty,\infty)$ provided that
$y_m(t)\in L^2(\mathbb{T})$ for $1\leq m\leq M$. Even if  $\beta_{max,\mathbf{N}}$ is not bounded over $\mathbf{N}$,  we can still get $\epsilon^2(x_r,\mathbf{N})\to 0$ provided that   the rate of $\abs{a_j(n)b_{mj}(n)}\to 0$  is  sufficiently fast for $ 1\leq j\leq R,1\leq m\leq M$. In fact,  for every example   in the paper, $\beta_{max,\mathbf{N}}$ is   bounded over $\mathbf{N}$. We give an upper bound of $\beta_{max,\mathbf{N}}$  for each example in Table \ref{T_ub_beta}. It is noted that the upper bounds presented in this table are not tight (implies that there may exist smaller upper bounds for each example), but it is sufficient to illustrate the convergence property.

 Another observation   reflected in Eq.  (\ref{upperbound}) is that  $\epsilon^2(x_r,\mathbf{N})= 0$  if $a_r(n)=0$ ($1\leq r\leq R$) for all $n\notin I^{\mathbf{N}}$.   It means that if $x_r(t)\in B_{\mathbf{N}}$ ($1\leq r\leq R$), whatever sampling scheme we select in the MIMO system, the reconstruction operator $\mathcal{T}_{\mathbf{N}}$ defined
by (\ref{reconst-formula}) will come into being the same result if the
total samples used in Eq. (29) is equal to $\mu(I^{\mathbf{N}})$. Otherwise,  $\mathcal{T}_{\mathbf{N}}$ will lead to imperfect  reconstruction for every $x_r(t)$ as long as  there is a input signal $x_{r_0}(t)\notin  B_{\mathbf{N}}$. Furthermore,  different sampling schemes will bring different errors under imperfect reconstruction situations even  the  number of samples is the same. Based on the error formula given in (\ref{exact_alias_error}), we will analyze  the errors of the MIMO type reconstructions for several specific sampling schemes   experimentally.

 \begin{table}[htb]
\caption{$\beta_{max,\mathbf{N}}$ is  bounded over $\mathbf{N}$ for all the examples, an  upper bound is given for each example. } \vspace{0.2cm}	\centering
	\begin{tabular}{c|c c c c c c c c}
\hline
No. of Example  & 1  & 2 & 3 &4 &5 &6 &7&8  \\
\hline
upper bound of $\beta_{max,\mathbf{N}}$  & $1$ & $2$ & $1$ & $4$ & $20$ &  $1$  & $4$
		& $2 \sqrt{2} $ \\ \hline
	\end{tabular}\label{T_ub_beta}
\end{table}

\section{Numerical Examples}\label{S_examples}

In this section, we use     numerical examples to show the effectiveness of the proposed   sampling and reconstruction framework. Firstly, we use the proposed FFT-based algorithm to   perfectly  reconstruct band-limited signals from their MIMO samples. Secondly,    the proposed MIMO framework is employed to sample and reconstruct  non-band-limited signals. We will perform the consistency test and the aliasing error analysis for the proposed reconstruction algorithm by numerical simulations.   Finally, we will examine the performance of the proposed reconstruction method in the presence of noise.  

We have presented  several examples of MIMO systems with left invertible $\mathbf{B}(n)$ in Table \ref{example_MIMOsyst}.  
Five representative MIMO systems (Example 1, 2, 5,  6, 7) are selected to conduct   numerical experiments to verify the theoretical results.
Based on the number of inputs and outputs of these MIMO systems, and noting that these systems involve derivatives or translations, we will refer to the corresponding    sampling schemes  as  \texttt{S-22d},  \texttt{S-22t},  \texttt{S-23t}, \texttt{S-23d} and \texttt{S-24d}, respectively.

\begin{figure}[!ht]
  \centering
  \includegraphics[width=0.95\textwidth]{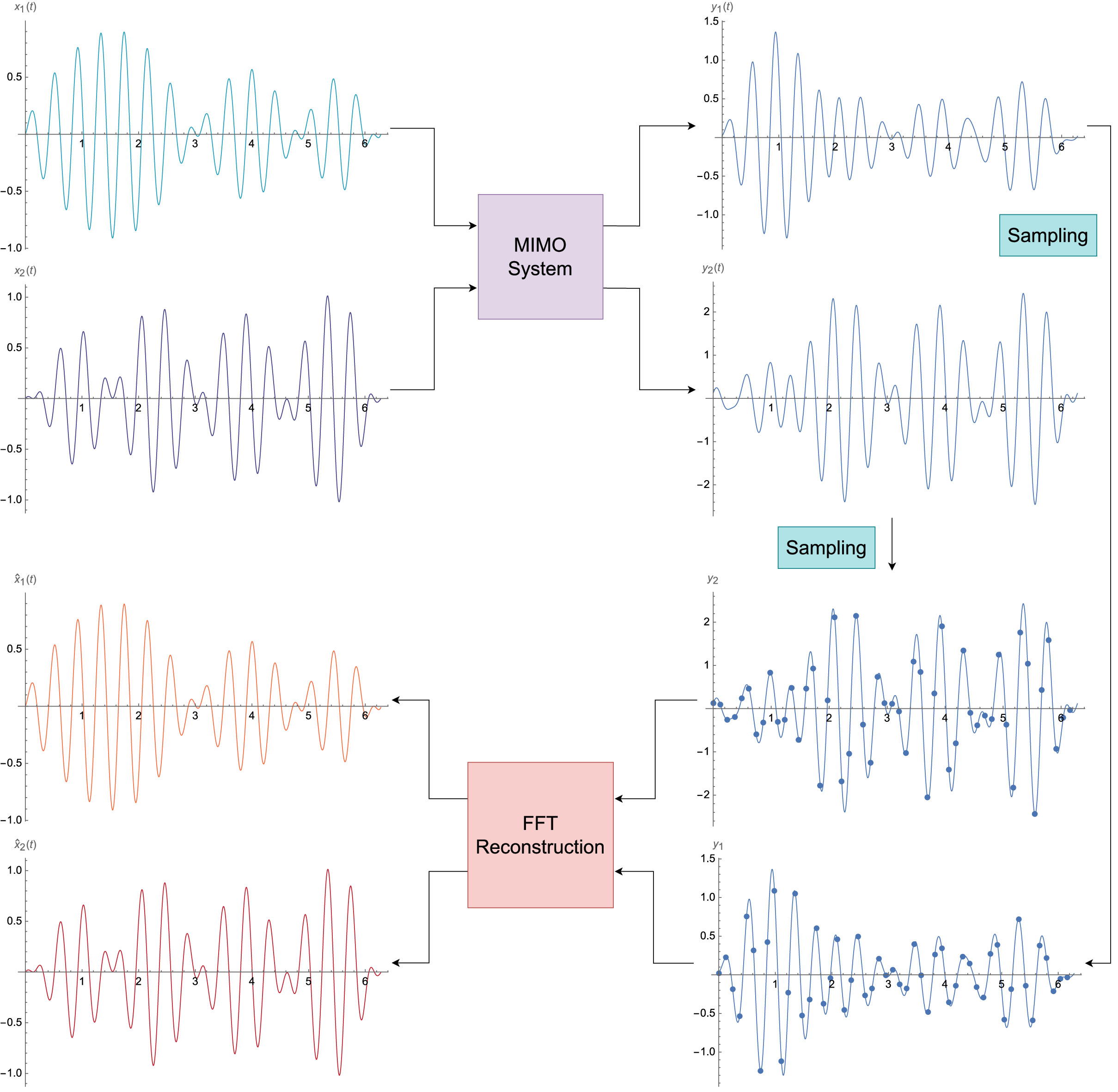}
  \caption{Error-free Reconstruction. The band-limited input signals  $x_1(t), x_2(t)$ (bandwidth $51$) are   sampled by the sampling scheme \texttt{S-22t} (see bottom right). They can be perfectly reconstructed by  the proposed FFT-based algorithm.}\label{MIMO_perfect_recons}
\end{figure}

\subsection{Error-free Reconstruction}\label{error-free-ex}
Let 
\begin{align*}
f_1(t)&=0.015 \left(0.12 t^4-1.28 t^3-5.88 t^2+e^{-t^2}-4.38 t+32.2325\right) \sin (15 t) \cos (1.5  -t),\\
f_2(t)&=0.03 \left(-0.3 t^4+1.2 t^3+1.6 t^2+2 e^{-t^2}-3 t+35\right) \sin (2 t) \cos (15 t),
\end{align*}
In this subsection, the test input signals $x_1(t), x_2(t)$ are constructed by passing   $f_1(t)$ and $f_2(t)$ through an ideal low-pass filter with bandwidth $51$. That is,
\begin{align}
x_1(t)&= (f_1*D_{25})(t), \label{bws1}\\
x_2(t)&=(f_1*D_{25})(t), \label{bws2} 
\end{align}
where $D_n(t)=\sum_{k=-n}^n e^{\qi k t}$ is the Dirichlet kernel. 

 A showcase of sampling and reconstruction processes for \texttt{S-22t}  is provided in Figure  \ref{MIMO_perfect_recons}. Concretely, the original input signals $x_1(t), x_2(t)$, the output signals $y_1(t), y_2(t)$, the MIMO samples and the reconstructed signals are presented in top-left, top-right, bottom-right and  bottom-left of Figure  \ref{MIMO_perfect_recons} respectively.   The reconstruction is indeed flawless, both in terms of error values and from a visual perspective.

Besides \texttt{S-22t}, we   have also taken other sampling schemes into account.  Without a doubt,  the original signals can also be perfectly reconstructed under the sampling schemes \texttt{S-22d}, \texttt{S-23d}, \texttt{S-23t} and  \texttt{S-24d}; nevertheless, the total number of samples required for reconstruction varies with different sampling schemes.  
In our setup,  the minimum total number of samples required for reconstruction is 
$$
 {M\cdot \mu(I^{\mathbf{N}})}/{\textstyle\lfloor\frac{M}{R} \rfloor}.
$$
This implies that for fixed $R$ and $\mu(I^{\mathbf{N}})$,   we need more samples to reconstruct the input signals when $M$ cannot be divided exactly by $R$. In the current setting (where $R=2, \mu(I^{\mathbf{N}})=51$),   we require $3\times 51=153$ samples to perfectly reconstruct  the input signals under the sampling scheme  \texttt{S-23t} and \texttt{S-23d}.  In contrast, under the  \texttt{S-22t},  \texttt{S-22d}  and  \texttt{S-24d} sampling schemes, only $102$ ($104$ for \texttt{S-24d}) samples are needed to achieve error-free signal reconstruction.

\begin{figure}[!ht]
  \centering
  \includegraphics[width=0.95\textwidth]{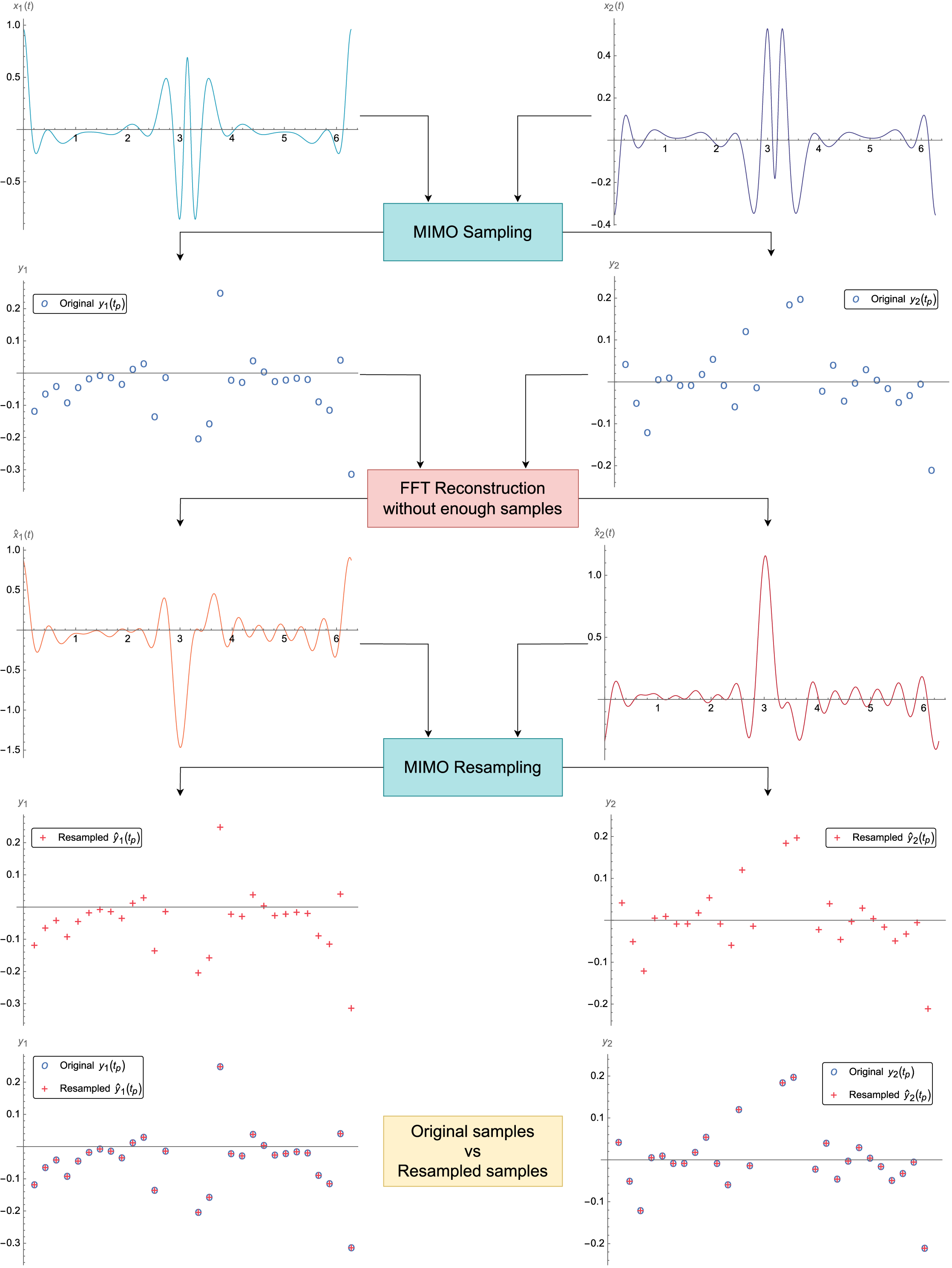}
  \caption{Consistency test under the sampling scheme \texttt{S-22t}.   Row 1--2:  the original signals and their MIMO samples. Row 3: the reconstructed signals. Row 4: the resampled samples. Row 5: the original and resampled samples are identical.}\label{consistency_test_1}
\end{figure}

\begin{figure}[!ht]
  \centering
  \includegraphics[width=0.95\textwidth]{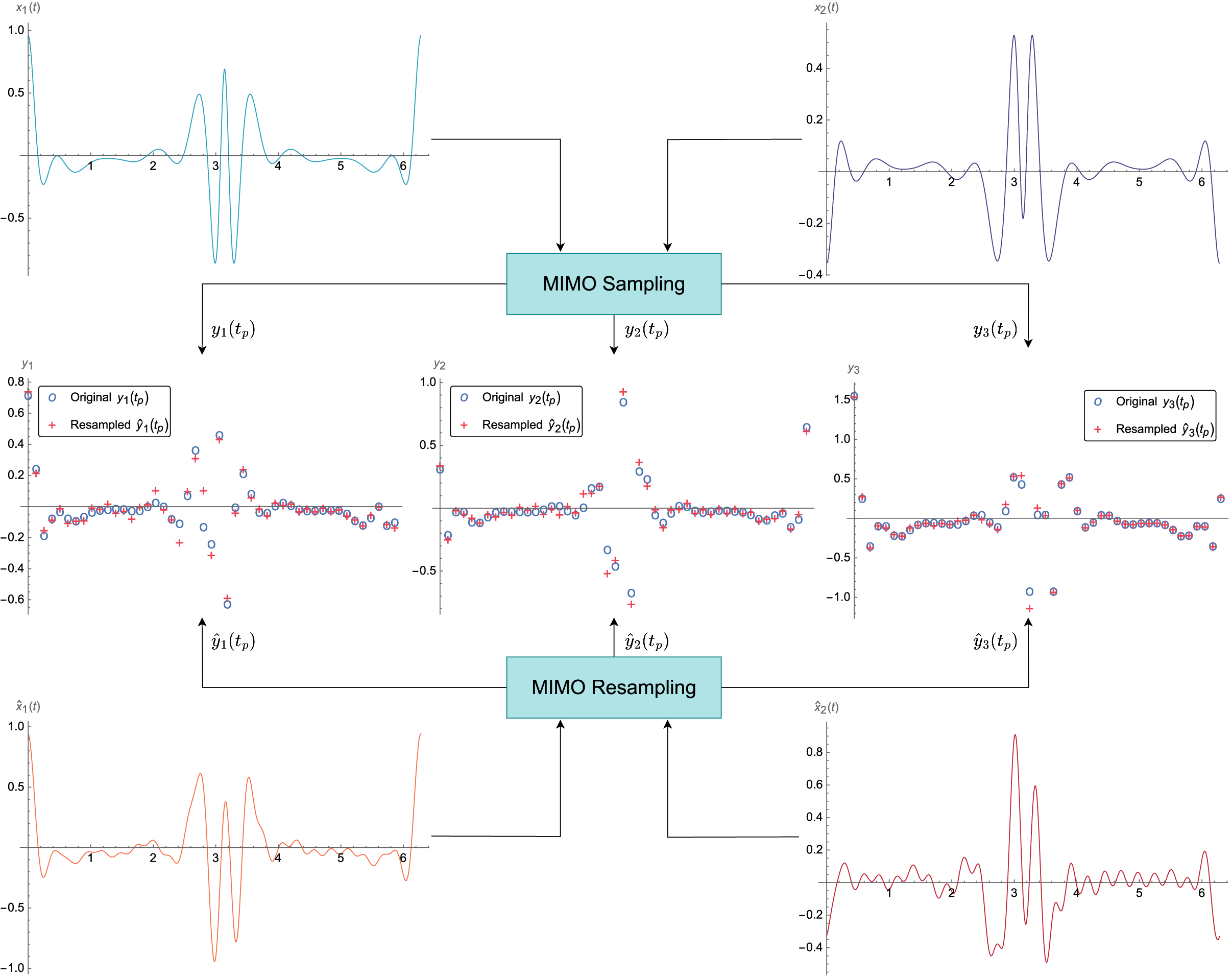}
  \caption{Consistency test under the sampling scheme \texttt{S-23t}.   Row 1:  the original signals. Row 3: the reconstructed signals. Row 2: The original samples do not match the resampled ones.}\label{consistency_test_2}
\end{figure}

\subsection{Consistency Test}
 Let
\begin{align*} 
&\phi_1(z)=\frac{0.08z^2+0.06z^{10}}{(1.3-z)(1.5-z)}+\frac{0.05z^3+0.
09z^{10}}{(1.2+z)(1.3+z)}, \\
&\phi_2(z)=\frac{0.036 z^{10}+0.024 z^2}{(z-1.3) (1.6  -z)}-\frac{0.06 \
z^{10}+0.048 z^3}{(z+1.2) (z+1.35)},
\end{align*}
then $\phi_1(z), \phi_2(z)$ belong to Hardy space on the  unit disk. In this subsection,  we use $x_1(t)=\operatorname{Re}[\phi_1(e^{\qi t})], x_2(t)=\operatorname{Re}[\phi_2(e^{\qi t})]$  as the test input signals, which are  not band-limited.

As depicted in Figure \ref{consistency_test_1},   we reconstruct the input signals with a small number of MIMO samples (\texttt{S-22t}), and it is clear that the reconstructed signals are far from the original signals. Nevertheless, when we resample the reconstructed signals,   the same samples are obtained as in the initial sampling. This is in line with the theoretical result:  the consistency of MIMO sampling holds when $M$   is divisible by $R$.

If  $M$   is not divisible by $R$, such as the sampling scheme \texttt{S-23t}, the consistency cannot be guaranteed. The experimental results show that for the sampling scheme \texttt{S-23t}, the resampled samples differ from the original ones, as depicted in Figure \ref{consistency_test_2}. It is noted that the inconsistency diminishes as the number of samples increases. This is because the increase in samples will lead to a reduction in reconstruction error, it is not surprising the resampling results will  increasingly approach the original samples. 

Next, we will  examine how the aliasing error is affected by various sampling schemes and how  the aliasing error varies with the number of samples.

\begin{figure}[!ht]
  \centering
  \includegraphics[width=0.95\textwidth]{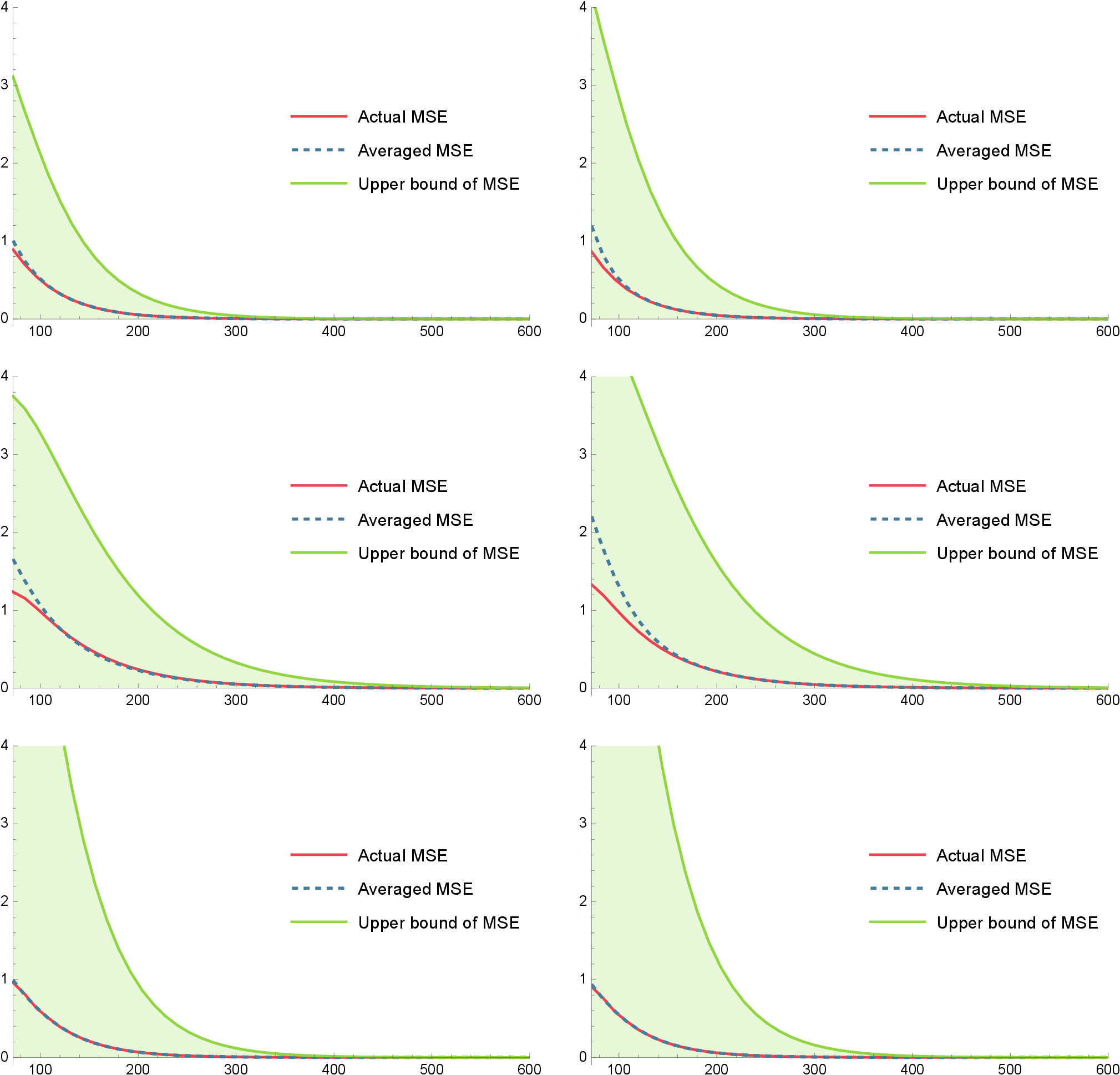}
  \caption{Actual MSE, averaged MSE and upper bound of MSE versus the total number of samples under the sampling scheme \texttt{S-22d} (first row), \texttt{S-23d} (second row) and  \texttt{S-24d} (third row). The test input signals are  $x_1(t)=\operatorname{Re}[\phi_1(e^{\qi t})], x_2(t)=\operatorname{Re}[\phi_2(e^{\qi t})]$. Left: the errors for reconstructing  $x_1(t)$. Right: the errors for reconstructing  $x_2(t)$.}\label{act_ave_up_error}
\end{figure} 

\subsection{Aliasing Error Analysis}

 In    Section \ref{S_consis_error}, we have provided a comprehensive theoretical analysis of the aliasing   error.  In addition to the averaged MSE, i.e.  $\epsilon(x_r,\mathbf{N})$,  we have also furnished  a more straightforward and easily computable upper bound  for $\epsilon(x_r,\mathbf{N})$.   
 In this part,  we are going to verify the following points experimentally: first, whether $\epsilon(x_r,\mathbf{N})$ approximates $\xi(x_r,\mathbf{N})$; second, whether  $\xi(x_r,\mathbf{N})$ is  dominated by   the provided error upper bound; and third, whether the actual MSE, the averaged MSE, and the  upper bound of MSE all converge to zero as $L\to \infty$.

We still use $x_1(t)=\operatorname{Re}[\phi_1(e^{\qi t})], x_2(t)=\operatorname{Re}[\phi_2(e^{\qi t})]$  as the test input signals.   The actual MSE, the averaged MSE, and the  upper bound of MSE for reconstructing $x_1(t)$ and $x_2(t)$ under the sampling scheme  \texttt{S-22d}, \texttt{S-23d} and  \texttt{S-24d}  are   plotted in Figure \ref{act_ave_up_error}. We see that the averaged MSE is close to the actual MSE and the two are nearly equivalent when the number of samples is slightly larger (especially for  \texttt{S-24d}).  Moreover,  both the averaged MSE and the actual MSE are less than the  upper bound of MSE. Although the provided upper bounds of MSE are rather rough (they are significantly larger than the actual MSEs), it is evident that   the  upper bounds still tend  toward $0$.

Figure \ref{error_compared} illustrates the actual MSE of the reconstruction  under   different sampling schemes and with varying numbers of samples.  It is noted that 
the horizontal  in this figure represents the number of samples per individual output channel, implying that sampling schemes with more channels will have a larger total number of samples. 
It is reasonable that when $L$ is fixed, the actual MSE of reconstruction under  the sampling scheme \texttt{S-24d} will be smaller than that under \texttt{S-23d} and \texttt{S-22d}. However,  the actual MSE of the reconstruction under the sampling scheme \texttt{S-23d} did not fall below that of \texttt{S-22d} as expected. We observe that when $L$ is fixed,   the actual MSE of the reconstruction  under the sampling scheme  \texttt{S-23d} is almost identical to that under the sampling scheme  \texttt{S-22d}; that is, the additional samples from the extra channel in  \texttt{S-23d} do not contribute to an improvement in accuracy. The main reason for this situation is that the additional samples do not effectively broaden the bandwidth of the reconstructed signals,   thus do not improve the reconstruction accuracy.

\begin{figure}[ht]
  \centering
  \includegraphics[width=0.95\textwidth]{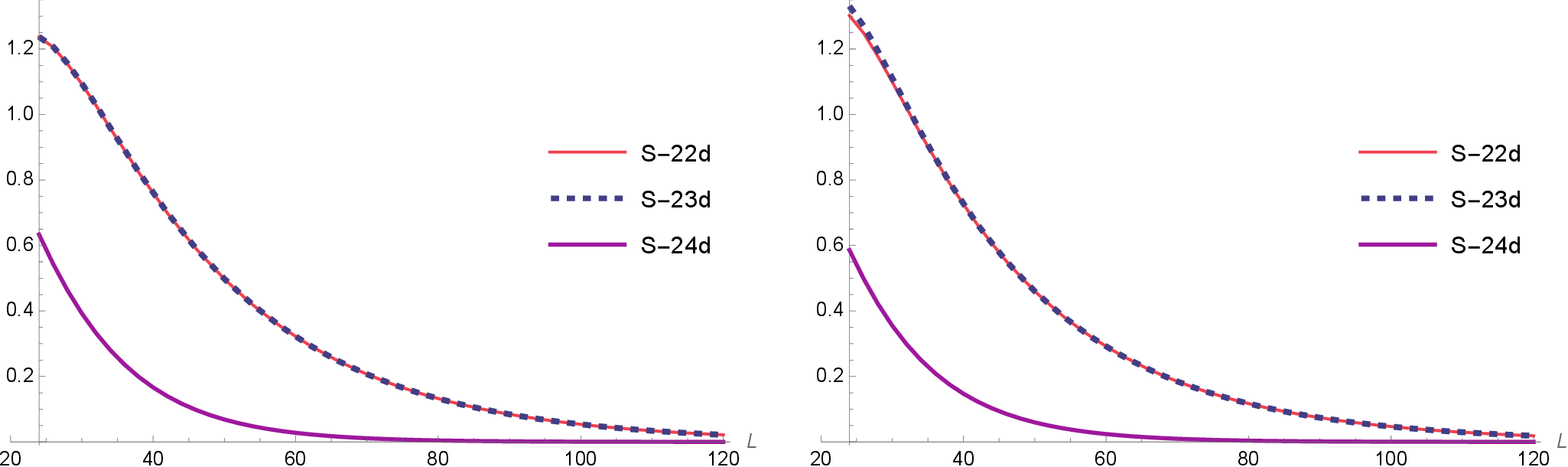}
  \caption{Actual MSE  versus $L$ (the number of samples in each channel) under  the sampling scheme  \texttt{S-22d}, \texttt{S-23d} and  \texttt{S-24d}. The test input signals are  $x_1(t)=\operatorname{Re}[\phi_1(e^{\qi t})], x_2(t)=\operatorname{Re}[\phi_2(e^{\qi t})]$. Left: the errors for reconstructing  $x_1(t)$. Right: the errors for reconstructing  $x_2(t)$.}\label{error_compared}
\end{figure}

In fact, the bandwidth of the reconstructed signals is 
$
\lfloor\tfrac{M}{R}\rfloor\cdot L,
$
where $L$ is the number of samples in each channel. Therefore, when \( M \) is not divisible by \( R \),  rounding down $\tfrac{M}{R}$ will result in a ``discount" to the bandwidth of the reconstructed signals. Given also that consistency is   not guaranteed in this scenario, it is not recommended to use a sampling scheme where \( R \) does not divide \( M \)  for converting between continuous and discrete signals.

 \begin{figure}[!ht]
  \centering
  \includegraphics[width=0.95\textwidth]{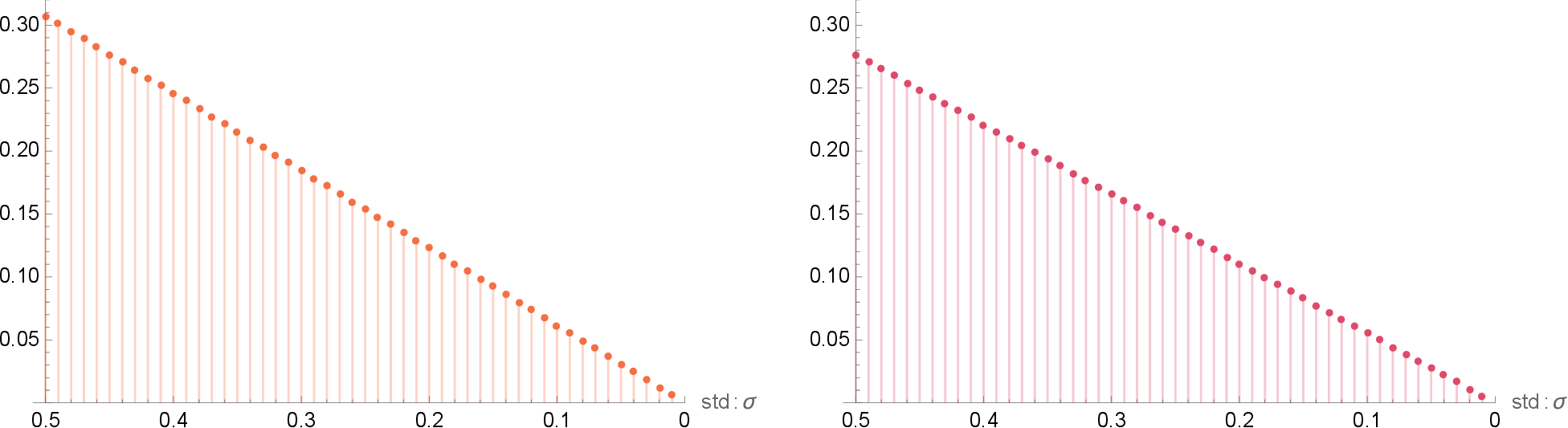}
  \caption{Actual MSE versus  $\sigma_{\eta}$ under the sampling scheme \texttt{S-22d} with $L=51$. The test input signals  $x_1(t)$ and  $x_2(t)$ are given by (\ref{bws1}) and  (\ref{bws2}) respectively. Left: the errors for reconstructing  $x_1(t)$. Right: the errors for reconstructing  $x_2(t)$.}  \label{stable_analysis_noi_sigma}
\end{figure} 

\begin{figure}[!ht]
  \centering
  \includegraphics[width=0.95\textwidth]{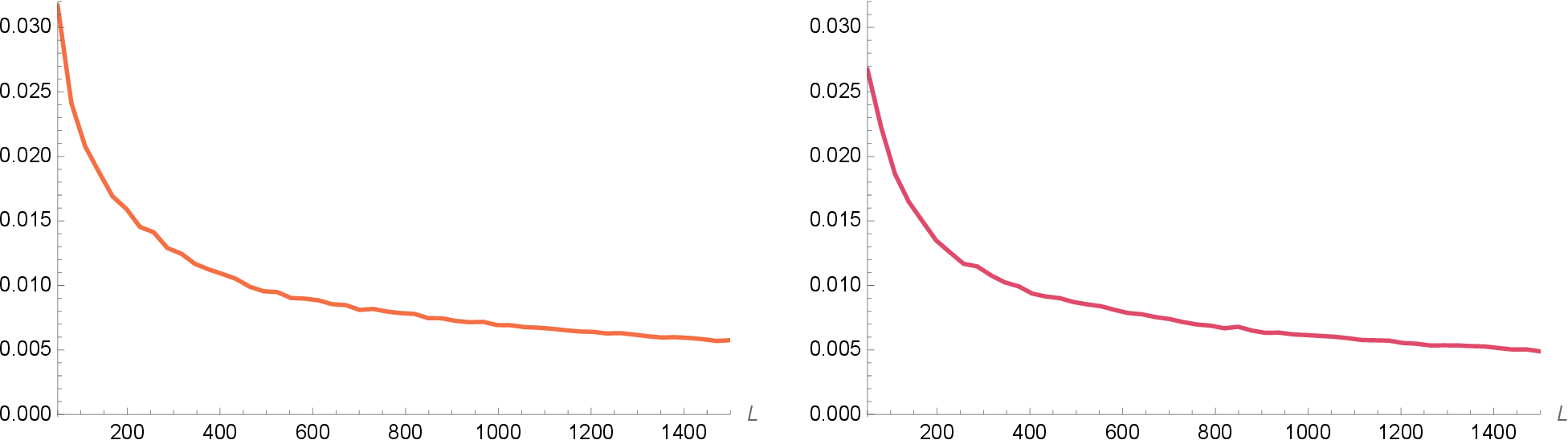}
  \caption{$\widetilde{\xi}(x_r,\mathbf{N})$ versus $L$  under the sampling scheme \texttt{S-22d} with $\sigma_{\eta}=0.05$. The test input signals  $x_1(t)$ and  $x_2(t)$ are given by (\ref{bws1}) and  (\ref{bws2}) respectively. Left: the errors for reconstructing  $x_1(t)$. Right: the errors for reconstructing  $x_2(t)$.}\label{converge_noi_L}
\end{figure}

\subsection{Signal Reconstruction from Noisy Samples}

In this subsection, we will analyze the behavior of MIMO sampling in the presence of noise.
Let (\ref{bws1}) and (\ref{bws2}) be the test input signals, we will use the noisy samples
\begin{equation*}
 s_{m,p}=y_m(\tfrac{2\pi p}{L})  +\eta_{m,p},\quad   0\leq p\leq L-1,~  1\leq m \leq M 
\end{equation*}
to reconstruct  $x_1(t)$ and $x_2(t)$. Here,  \(\eta_{m,p}\) represents an independent and identically distributed (IID)  noise process characterized by a zero mean and a variance of \(\sigma_{\eta}^2\).  We aim to demonstrate the following two assertions.  First, the proposed reconstruction method is stable with respect to noise, meaning that small errors in samples caused by noise will not lead to significant errors in the reconstructed signals; second, oversampling helps to reduce the impact of noise on the reconstruction.

We can see from Figure \ref{stable_analysis_noi_sigma} that the actual MSE is linearly related to the noise level (standard deviation). Furthermore, when the standard deviation approaches $0$, the error also tends to $0$. This implies that the proposed reconstruction method is stable against noise. For any fixed $t$, it follows from (\ref{reconst-formula}) that $[\mathcal{T}_{\mathbf{N}} x_r](t)$ can just regarded as  a finite weighted sum of all samples. Therefore, the reconstructed signals are continuous with respect to the  samples.

 Incorporating oversampling with low-pass post-filtering is a feasible tactic \cite{Pawlak2003postfilter} to reduce the errors  caused by noise in  signal reconstruction.
To validate   this tactic continues to be effective in the MIMO configuration, we   will investigate how   $\widetilde{\xi}(x_r,\mathbf{N})$ changes with $L$, where
$$
\widetilde{\xi}(x_r,\mathbf{N}) = \sqrt{ \frac{1}{2\pi}\int_{\mathbb{T}}\abs{x_r(t) -\left([\widetilde{\mathcal{T}}_{\mathbf{N}} x_r]*D_{26}\right)(t)}^2 dt  },
$$
and 
$$
[\widetilde{\mathcal{T}}_{\mathbf{N}} x_r](t) = \frac{1}{L} \sum_{m=1}^M \sum_{p=0}^{L-1}  s_{m,p}~ g_{rm}(t-\tfrac{2\pi p}{L}),\quad  1\leq r \leq R.
$$
From Figure \ref{converge_noi_L},  we see that as \( L \) increases, $\widetilde{\xi}(x_r,\mathbf{N})$  gradually approaches zero. More precisely, it is observed that $\widetilde{\xi}(x_r,\mathbf{N})$ decreases more rapidly in the initial phase, but the rate of decrease slows considerably in the subsequent segment. Therefore, the combination of oversampling and low-pass filtering can suppress the impact of noise in the framework of MIMO reconstruction. However, when there are a sufficient number of samples, simply truncating high frequencies does not yield satisfactory results; more effective noise reduction methods should be explored. This topic will be  investigated in our  upcoming studies.

\section{Conclusion}\label{S_conclusion}

In this paper, we present a MIMO sampling and reconstruction framework for signals of finite duration. We demonstrate that under certain conditions on the MIMO system and the sampling density, the input periodic band-limited signals can be perfectly reconstructed from the samples of the output signals. Moreover, the reconstruction can be performed stably and efficiently using an FFT-based algorithm.

We also explore the performance of the proposed sampling theorem in reconstructing non-band-limited signals. It is shown that if $M$ is divisible by $R$, then the proposed MIMO sampling is consistent. Otherwise, consistency is not guaranteed, and there is a loss of bandwidth in the reconstructed signals. Regardless of whether $M/R$ is an integer, we provide the average MSE to measure the difference between the original input signals and the reconstructed signals. The experiments show that the average MSE closely approximates the actual MSE.

Additionally, when the output signals are sampled in the presence of additive noise, we show that the proposed reconstruction method remains stable. Furthermore, the effectiveness of post-filtering in noise reduction is also investigated.

\section*{Acknowledgement}

This work was supported by the Guangdong Basic and Applied Basic Research Foundation (No. 2019A1515111185), the Research Development Foundation of Wenzhou Medical University (QTJ18012), Wenzhou Science and Technology Bureau (G2020031) and Scientific Research Task of Department of Education of Zhejiang Province (Y202147071).

\section*{Appendix}

\subsection*{Upper bound of $\beta_{max,\mathbf{N}}$}

We only provide the detail to derive an upper bound of $\beta_{max,\mathbf{N}}$ for Example 5.

$$
\mathbf{Q}(n) =
  \frac{1}{-8 e^{\qi \alpha  n}+3 e^{2 \qi \alpha  n}+e^{4 \qi \alpha  n}+6}
\begin{bmatrix}
 -2 e^{\qi \alpha  n}-e^{2 \qi \alpha  n}+2 & -2+e^{3 \qi \alpha  n} & 2 \left(-e^{\qi \alpha  n}+e^{2 \qi \alpha  n}+1\right) \\
 3 e^{\qi \alpha  n}+e^{3 \qi \alpha  n}-2 & -2 e^{\qi \alpha  n}-e^{2 \qi \alpha  n}+5 & -4 e^{\qi \alpha  n}+e^{2 \qi \alpha  n}+1 \\
\end{bmatrix}
$$
Let $f(z)=-8z+3z^2+z^4+6$, then $f(z)$ has four zeros in complex plane $\mathbb{C}$ and none  of them lie within the unit circle. Thus there exists $\varepsilon_0 >0$ such that
$\abs{f(z)}>\varepsilon_0$ for all $z\in \{z:\abs{z}=1\}$. In fact, by solving the   optimization problem
$$\min~\abs{f(z)} \quad \textrm{s. t.} \quad \abs{z}=1,$$
we have that $\abs{f(z)}>\frac{2}{5}$ for all $z\in \{z:\abs{z}=1\}$. Note that no fraction in $\mathbf{Q}(n)$   has a numerator whose norm is greater than    $8$ (e.g., $\abs{-2 e^{\qi \alpha  n}-e^{2 \qi \alpha  n}+5}\leq 8$). It follows that
$\beta_{max,\mathbf{N}}\leq 8\div \frac{2}{5}=20$.  It generally does not happen that the numerator reaches the maximum value and the denominator reaches the minimum value at the same time. Therefore it  is likely that there is an upper bound smaller than $20$. Nevertheless, this is sufficient to illustrate the convergence.

{\small
 \setlength{\itemsep}{-0.8mm}
\bibliographystyle{IEEEtran}
\bibliography{IEEEabrv,myreference20240106p}

\begin{thebibliography}{10}
\providecommand{\url}[1]{#1}
\csname url@samestyle\endcsname
\providecommand{\newblock}{\relax}
\providecommand{\bibinfo}[2]{#2}
\providecommand{\BIBentrySTDinterwordspacing}{\spaceskip=0pt\relax}
\providecommand{\BIBentryALTinterwordstretchfactor}{4}
\providecommand{\BIBentryALTinterwordspacing}{\spaceskip=\fontdimen2\font plus
\BIBentryALTinterwordstretchfactor\fontdimen3\font minus
  \fontdimen4\font\relax}
\providecommand{\BIBforeignlanguage}[2]{{%
\expandafter\ifx\csname l@#1\endcsname\relax
\typeout{** WARNING: IEEEtran.bst: No hyphenation pattern has been}%
\typeout{** loaded for the language `#1'. Using the pattern for}%
\typeout{** the default language instead.}%
\else
\language=\csname l@#1\endcsname
\fi
#2}}
\providecommand{\BIBdecl}{\relax}
\BIBdecl

\bibitem{liu2017signal}
N.~Liu, R.~Tao, R.~Wang, Y.~Deng, N.~Li, and S.~Zhao, ``Signal reconstruction
  from recurrent samples in fractional {Fourier} domain and its application in
  multichannel {SAR},'' \emph{Signal Process.}, vol. 131, pp. 288--299, 2017.

\bibitem{Li2019multichannel}
Y.~Li and Y.~Bresler, ``Multichannel sparse blind deconvolution on the
  sphere,'' \emph{{IEEE} Trans. Inf. Theory}, vol.~65, no.~11, pp. 7415--7436,
  2019.

\bibitem{Chen2013shannon}
Y.~{Chen}, Y.~C. {Eldar}, and A.~J. {Goldsmith}, ``Shannon meets {Nyquist}:
  Capacity of sampled {Gaussian} channels,'' \emph{IEEE Trans. Inf. Theory},
  vol.~59, no.~8, pp. 4889--4914, 2013.

\bibitem{Solodky2021optimal}
G.~Solodky and M.~Feder, ``Optimal sampling of a bandlimited noisy multiple
  output channel,'' \emph{IEEE Trans. Signal Process.}, vol.~69, pp.
  2026--2041, 2021.

\bibitem{papoulis1977generalized}
A.~Papoulis, ``Generalized sampling expansion,'' \emph{{IEEE} Trans. Circuits
  Syst.}, vol.~24, no.~11, pp. 652--654, 1977.

\bibitem{Kang2010asymmetric}
S.~Kang, J.~Kim, and K.~Kwon, ``Asymmetric multi-channel sampling in shift
  invariant spaces,'' \emph{J. Math. Anal. Appl.}, vol. 367, no.~1, pp. 20--28,
  2010.

\bibitem{xu2017multichannel}
L.~Xu, R.~Tao, and F.~Zhang, ``Multichannel consistent sampling and
  reconstruction associated with linear canonical transform,'' \emph{IEEE
  Signal Process. Lett.}, vol.~24, no.~5, pp. 658--662, 2017.

\bibitem{wei2019convolution}
D.~Wei and Y.~M. Li, ``Convolution and multichannel sampling for the offset
  linear canonical transform and their applications,'' \emph{IEEE Trans. Signal
  Process.}, vol.~67, no.~23, pp. 6009--6024, 2019.

\bibitem{Azhar2023papoulis}
A.~Y. Tantary, F.~A. Shah, and A.~I. Zayed, ``Papoulis' sampling theorem:
  Revisited,'' \emph{Appl. Comput. Harmon. Anal.}, vol.~64, pp. 118--142, 2023.

\bibitem{zhao2018generalized}
H.~Zhao, L.~Qiao, N.~Fu, and G.~Huang, ``A generalized sampling model in
  shift-invariant spaces associated with fractional {F}ourier transform,''
  \emph{Signal Process.}, vol. 145, pp. 1--11, 2018.

\bibitem{akhondi2010multichannel}
H.~Akhondi~Asl, P.~L. Dragotti, and L.~Baboulaz, ``Multichannel sampling of
  signals with finite rate of innovation,'' \emph{IEEE Signal Process. Lett.},
  vol.~17, no.~8, pp. 762--765, 2010.

\bibitem{Pohl2012U-invariant}
V.~Pohl and H.~Boche, ``U-invariant sampling and reconstruction in atomic
  spaces with multiple generators,'' \emph{IEEE Trans. Signal Process.},
  vol.~60, no.~7, pp. 3506--3519, 2012.

\bibitem{Prender2006minimum}
R.~Prendergast and T.~Nguyen, ``Minimum mean-squared error reconstruction for
  generalized undersampling of cyclostationary processes,'' \emph{IEEE Trans.
  Signal Process.}, vol.~54, no.~8, pp. 3237--3242, 2006.

\bibitem{eldar_2015}
Y.~C. Eldar, \emph{Sampling Theory: Beyond Bandlimited Systems}.\hskip 1em plus
  0.5em minus 0.4em\relax Cambridge: Cambridge University Press, 2015.

\bibitem{medina2018papoulis}
J.~M. Medina and B.~Cernuschi-Frías, ``On the papoulis sampling theorem: Some
  general conditions,'' \emph{IEEE Trans. Inf. Theory}, vol.~64, no.~10, pp.
  6722--6730, 2018.

\bibitem{Cheng2022signal}
D.~Cheng, X.~Hu, and K.~I. Kou, ``Signal reconstruction from noisy multichannel
  samples,'' \emph{Digit. Signal Prog.}, vol. 129, p. 103673, 2022.

\bibitem{Yu2010MIMO}
Y.~Yu, A.~P. Petropulu, and H.~V. Poor, ``{MIMO} radar using compressive
  sampling,'' \emph{IEEE J. Sel. Topics Signal Process.}, vol.~4, no.~1, pp.
  146--163, 2010.

\bibitem{wang2022nonlinear}
S.~Wang, M.~He, and Y.~Liu, ``Nonlinear {MIMO} communication with
  $\pi$-periodic phase measurements,'' \emph{IEEE Trans. Wireless Commun.},
  vol.~21, no.~7, pp. 4856--4870, 2022.

\bibitem{Seidner1998intro}
D.~Seidner, M.~Feder, D.~Cubanski, and S.~Blackstock, ``Introduction to vector
  sampling expansion,'' \emph{IEEE Signal Process. Lett.}, vol.~5, no.~5, pp.
  115--117, 1998.

\bibitem{Seidner2000vector}
D.~Seidner and M.~Feder, ``Vector sampling expansion,'' \emph{IEEE Trans.
  Signal Process.}, vol.~48, no.~5, pp. 1401--1416, 2000.

\bibitem{Venka2003sampling}
R.~Venkataramani and Y.~Bresler, ``Sampling theorems for uniform and periodic
  nonuniform {MIMO} sampling of multiband signals,'' \emph{IEEE Trans. Signal
  Process.}, vol.~51, no.~12, pp. 3152--3163, 2003.

\bibitem{Venka2003filter}
------, ``Filter design for {MIMO} sampling and reconstruction,'' \emph{IEEE
  Trans. Signal Process.}, vol.~51, no.~12, pp. 3164--3176, 2003.

\bibitem{venka2004multiple}
------, ``Multiple-input multiple-output sampling: necessary density
  conditions,'' \emph{IEEE Trans. Inf. Theory}, vol.~50, no.~8, pp. 1754--1768,
  2004.

\bibitem{Ma2023nonuniform}
J.~Ma, G.~Li, R.~Tao, and Y.~Li, ``Nonuniform {MIMO} sampling and
  reconstruction of multiband signals in the fractional {Fourier} domain,''
  \emph{IEEE Signal Process. Lett.}, vol.~30, pp. 653--657, 2023.

\bibitem{Sharma2011vector}
K.~K. Sharma, ``Vector sampling expansions and linear canonical transform,''
  \emph{IEEE Signal Process. Lett.}, vol.~18, no.~10, pp. 583--586, 2011.

\bibitem{Feuer2006generalization}
A.~Feuer, G.~Goodwin, and M.~Cohen, ``Generalization of results on vector
  sampling expansion,'' \emph{IEEE Trans. Signal Process.}, vol.~54, no.~10,
  pp. 3805--3814, 2006.

\bibitem{shang2007vector}
Z.~Shang, W.~Sun, and X.~Zhou, ``Vector sampling expansions in shift invariant
  subspaces,'' \emph{J. Math. Anal. Appl.}, vol. 325, no.~2, pp. 898--919,
  2007.

\bibitem{Kim2008vector}
J.~Kim and K.~Kwon, ``Vector sampling expansion in {Riesz} bases setting and
  its aliasing error,'' \emph{Appl. Comput. Harmon. Anal.}, vol.~25, no.~3, pp.
  315--334, 2008.

\bibitem{jacob2002sampling}
M.~Jacob, T.~Blu, and M.~Unser, ``Sampling of periodic signals: A quantitative
  error analysis,'' \emph{IEEE Trans. Signal Process.}, vol.~50, no.~5, pp.
  1153--1159, 2002.

\bibitem{margolis2008nonuniform}
E.~Margolis and Y.~C. Eldar, ``Nonuniform sampling of periodic bandlimited
  signals,'' \emph{IEEE Trans. Signal Process.}, vol.~56, no.~7, pp.
  2728--2745, 2008.

\bibitem{xiao2013sampling}
L.~Xiao and W.~Sun, ``Sampling theorems for signals periodic in the linear
  canonical transform domain,'' \emph{Opt. Commun.}, vol. 290, pp. 14--18,
  2013.

\bibitem{Mohammadi2018sampling}
E.~Mohammadi and F.~Marvasti, ``Sampling and distortion tradeoffs for
  bandlimited periodic signals,'' \emph{IEEE Trans. Inf. Theory}, vol.~64,
  no.~3, pp. 1706--1724, 2018.

\bibitem{cheng2019fft}
D.~Cheng and K.~I. Kou, ``{FFT} multichannel interpolation and application to
  image super-resolution,'' \emph{Signal Process.}, vol. 162, pp. 21--34, 2019.

\bibitem{Cheng2020multi}
------, ``Multichannel interpolation of nonuniform samples with application to
  image recovery,'' \emph{J. Comput. Appl. Math.}, vol. 367, p. 112502, 2020.

\bibitem{fraser1989interpolation}
D.~Fraser, ``Interpolation by the {FFT} revisited-an experimental
  investigation,'' \emph{IEEE Trans. Acoust. Speech Signal Process.}, vol.~37,
  no.~5, pp. 665--675, 1989.

\bibitem{selva2015fft}
J.~Selva, ``{FFT} interpolation from nonuniform samples lying in a regular
  grid,'' \emph{IEEE Trans. Signal Process.}, vol.~63, no.~11, pp. 2826--2834,
  2015.

\bibitem{folland1992fourier}
G.~B. Folland, \emph{Fourier analysis and its applications}.\hskip 1em plus
  0.5em minus 0.4em\relax American Mathematical Soc., 1992.

\bibitem{Boehme1975generalized}
T.~K. Boehme and G.~Wygant, ``Generalized functions on the unit circle,''
  \emph{Am. Math. Mon.}, vol.~82, no.~3, pp. 256--261, 1975.

\bibitem{Unser1994nonideal}
M.~Unser and A.~Aldroubi, ``A general sampling theory for nonideal acquisition
  devices,'' \emph{IEEE Trans. Signal Process.}, vol.~42, no.~11, pp.
  2915--2925, 1994.

\bibitem{Unset1997generalized}
M.~Unser and J.~Zerubia, ``Generalized sampling: stability and performance
  analysis,'' \emph{IEEE Trans. Signal Process.}, vol.~45, no.~12, pp.
  2941--2950, 1997.

\bibitem{Hira2007consistent}
A.~Hirabayashi and M.~Unser, ``Consistent sampling and signal recovery,''
  \emph{IEEE Trans. Signal Process.}, vol.~55, no.~8, pp. 4104--4115, 2007.

\bibitem{poon2014consistent}
C.~Poon, ``A consistent and stable approach to generalized sampling,'' \emph{J.
  Fourier Anal. Appl.}, vol.~20, no.~5, pp. 985--1019, 2014.

\bibitem{Pawlak2003postfilter}
M.~{Pawlak}, E.~{Rafajlowicz}, and A.~{Krzyzak}, ``Postfiltering versus
  prefiltering for signal recovery from noisy samples,'' \emph{IEEE Trans. Inf.
  Theory}, vol.~49, no.~12, pp. 3195--3212, 2003.

\end{thebibliography}
}

\end{document}